\documentclass[10pt,journal]{IEEEtran}
\usepackage{ifpdf}
\usepackage{manfnt}
\usepackage[mathscr]{eucal}
\usepackage{graphicx}

\usepackage{algorithmicx}
\usepackage[ruled,vlined,commentsnumbered]{algorithm2e}
\usepackage{amssymb}
\usepackage{theorem}
\usepackage{caption}
\usepackage{subcaption}
\usepackage{cite}
\usepackage{color}
\usepackage{colortbl}
\definecolor{colorref}{rgb}{0.4648,0,0} 
\definecolor{colorcite}{rgb}{0,0.2902,0.1765}

\usepackage{bm}

\usepackage{amsmath}
\newtheorem{proposition}{Proposition}

\newtheorem{lemma}{Lemma}

\newtheorem{Cor}{Corollary}

\newcommand{\setn}{\mathcal{N}}
\newcommand{\setm}{\mathcal{M}}

\newcommand{\Rmnum}[1]{\uppercase\expandafter{\romannumeral #1}}

\renewcommand{\baselinestretch}{1.33}

\newcommand{\signal}[1]{{\boldsymbol{#1}}}

\newcommand{\real}{{\mathbb R}}

\newtheorem{definition}{Definition}
\newtheorem{remark}{Remark}
\newtheorem{fact}{Fact}

\newtheorem{problem}{Problem}
\newtheorem{example}{Example}
\newcommand{\Natural}{{\mathbb N}}
\newcommand{\refeq}[1]{(\ref{#1})}

\title{Connections between spectral properties of asymptotic mappings and solutions to wireless network problems}
\author{R.~L.~G.~Cavalcante $\dagger$\thanks{This work was supported by a Nokia University Donation.}, {\it Member, IEEE}, Qi Liao $\ddagger$, {\it Member, IEEE}, and S.~Sta\'nczak $\dagger$, {\it Senior Member, IEEE} \\ 
	$\dagger$ Fraunhofer Heinrich Hertz Institute and Technical University of Berlin \\   $\ddagger$ Nokia Bell Labs, Stuttgart}

\newenvironment{proof}{{\it Proof:}}{\hfill$\square$\\}

\begin{document}
\maketitle

\begin{abstract}

	In this study we establish connections between asymptotic functions and properties of solutions to important problems in wireless networks. We start by introducing a  class of self-mappings (called asymptotic mappings) constructed with asymptotic functions, and we show that spectral properties of these mappings explain the behavior of solutions to some max-min utility optimization problems. For example, in a common family of max-min utility power control problems, we prove that the optimal utility as a function of the power available to transmitters is approximately linear in the low power regime. However, as we move away from this regime, there exists a transition point, easily computed from the spectral radius of an asymptotic mapping, from which gains in utility become increasingly marginal. From these results we derive analogous properties of the transmit energy efficiency. In this study we also generalize and unify existing approaches for feasibility analysis in wireless networks. Feasibility problems often reduce to determining the existence of the fixed point of a standard interference mapping, and we show that the spectral radius of an asymptotic mapping provides a necessary and sufficient  condition for the existence of such a fixed point. We further present a result that determines whether the fixed point satisfies a constraint given in terms of a monotone norm.  
		
\end{abstract}


\section{Introduction}

Asymptotic functions play a prominent role in nonlinear analysis \cite{aus03,rock70}, an area widely used to address important problems in wireless communications. However, with few exceptions \cite{renato2016,renato2016maxmin}, asymptotic functions rarely appear in studies on wireless networks, so these functions have not been fully exploited in the context of wireless systems. A fact that may partially explain the absence of asymptotic functions in studies on wireless networks is the difficulty of obtaining simple analytic representations of these functions if we depart from the field of convex analysis. However, as we prove in this study, the same analytic simplification that is applied to compute asymptotic functions associated with convex functions can also be used to compute asymptotic functions associated with, among others, the well-known standard interference functions  \cite{yates95}, which have been extensively used in applications in wireless networks  \cite{martin11,slawomir09,renato2016,renato2016maxmin,renato2016power,feh2013,nuzman07,ho2015,qi2017}. The problems being solved in these applications can often be formalized as fixed point problems, feasibility problems, or optimization problems involving (self)-mappings constructed with standard interference functions. By computing the asymptotic function of each coordinate of these mappings, we introduce the hereafter called asymptotic (self-)mappings, which are shown to belong to a class of mappings that has received a great deal of attention from the mathematical literature \cite{gau04,gun95,krause1986perron,krause01,nussbaum1986convexity,lem13,krause2015}. 

Asymptotic mappings set the stage for the next main contributions of this study. In particular, we show that asymptotic mappings are useful to obtain rigorous insights into properties of the solutions to some problems in wireless networks, and we note that some of these properties have been proved on a case-by-case basis or only observed  in simulations. For example, building upon the results in \cite{nuzman07}, we use the concept of (nonlinear) spectral radius (see Definition~\ref{def.nl_radius}) of asymptotic mappings to derive upper bounds for both the utility and the efficiency (e.g, rate over power) of solutions to some  max-min utility optimization problems. These bounds are simple functions of the budget $\bar{p}$ (e.g., the power) available to transmitters, and they are asymptotically tight. Furthermore, all constants are known, so, in standard power control problems, the bounds are particularly useful to determine a transition point (e.g., a power budget) that gives a coarse indication of whether a wireless network is operating in a noise-limited regime or in an interference-limited regime. Moreover, they reveal that the network utility and the efficiency scale as $\Theta(1)$ and $\Theta(1/\bar{p})$, respectively, as $\bar{p}\to\infty$ (see Sect.~\ref{sect.preliminaries} for the formal definition of the big-$\Theta$ notation). Analogously, as $\bar{p}\to 0^+$, the bounds show that the network utility and the efficiency scale as $\Theta(\bar{p})$ and $\Theta(1)$, respectively.

As a further application of asymptotic mappings, we demonstrate that their spectral radius can be used to unify and generalize existing results in \cite{slawomir09,martin11,siomina12,renato2016,ho2014data,chiang2008power} related to feasibility of network designs. Formally, feasibility studies typically amount to determining whether a standard interference mapping has a fixed point, and previous studies have obtained necessary and sufficient conditions by assuming the interference mappings to be affine \cite[Ch.~2]{slawomir09}\cite[Sect.~2.5.2]{chiang2008power}, to be concave \cite{renato2016}, or to take a very special form \cite{siomina12,ho2014data}. Here we remove all these assumptions. We derive a simple necessary and sufficient  condition for the existence of the fixed point of an arbitrary (continuous) standard interference mapping. By doing so, previous results in the literature can be seen as corollaries of our result, which can also be used in the  analysis of existing interference models that do not satisfy the assumptions required by previous mathematical tools. In addition, to study the feasibility of network designs with power constraints, we show a simple result that provides us with information about the location of the fixed point. All the theory developed here is illustrated in network problems involving the well-known load coupled interference model investigated in \cite{Majewski2010,renato14SPM,renato2016power,siomina12,ho2014data,ho2015,feh2013}, among other studies.\footnote{A conference version of this study appeared in \cite{renato2017SIP}, and some proofs have been partially available in the accompanying unpublished technical reports in \cite{renato2017performance,renatomaxmin}. In particular, the latter reference shows similar bounds on the utility and efficiency of solutions to utility max-min problems. However, unlike the bounds shown in the next sections, those in \cite{renatomaxmin} are not necessarily asymptotically tight, and they are only valid for concave mappings. Furthermore, the results in \cite{renatomaxmin} do not establish  connections between the spectral radius of nonlinear mappings and problems in network max-min utility optimization and feasibility analysis.}

This study is structured as follows. In Sect.~\ref{sect.preliminaries} we introduce notation, review existing mathematical results in the literature, obtain the first main technical contribution of this study (Proposition~\ref{prop.af_properties}), and prove auxiliary results necessary for the following sections. Sect.~\ref{sect.bounds} and Sect.~\ref{sect.existence} provide novel insights into properties of general classes of problems common in wireless networks. Concrete applications of the theory developed here are shown in Sect.~\ref{sect.applications}.

\section{Mathematical preliminaries}
\label{sect.preliminaries}

The objective of this section is to develop part of the mathematical machinery required for the applications in this study. We further introduce notation and results in mathematics that are necessary to keep the presentation as self-contained as possible. The main technical contribution of this section is  Proposition~\ref{prop.af_properties}, which shows properties of asymptotic functions (Definition~\ref{def.asymp_func}) associated with classes of interference functions (Definition~\ref{def.inter_func}) that are common in problems in wireless networks. These properties enable us to introduce a simple self-mapping, hereafter called {asymptotic mapping (Definition~\ref{def.amap}),  associated with classes of self-mappings that have been used to analyze and optimize wireless networks for many years. As shown later in this study, spectral properties of asymptotic mappings lead to results that unify existing approaches in the literature (see, for example, Proposition~\ref{proposition.existence} and Proposition~\ref{proposition.feasibility_constraints}), and they also provide us with valuable insights into the behavior of wireless networks (see Sect.~\ref{sect.bounds}).

We now turn our attention to the standard notation and definitions used in this study. By $\real_+$ and $\real_{++}$ we denote the sets of, respectively, nonnegative reals and positive reals. The (effective) \emph{domain} of a function $f:\real^N\to \real\cup\{-\infty,\infty\}$ is given by $\mathrm{dom}~f:=\{\signal{x}\in\real^N~|~f(\signal{x})<\infty\}$, and $f$ is \emph{proper} if $\mathrm{dom}~f \neq\emptyset$ and $(\forall\signal{x}\in\real^N)~ f(\signal{x}) > -\infty$.  A function $f:\real^N\to \real\cup\{-\infty,\infty\}$ is \emph{positively homogeneous} if $\signal{0}\in\mathrm{dom}~f$ and $(\forall \signal{x}\in\real^N)(\forall \alpha> 0)f(\alpha\signal{x})=\alpha~f(\signal{x})$. We say that a proper function $f:\real^N\to\real\cup\{\infty\}$ is \emph{continuous if restricted to $C\subset \mathrm{dom}~f \subset \real^N$}  if $(\forall\signal{x}\in C) (\forall(\signal{x}_n)_{n\in\Natural}\subset C) \lim_{n\to\infty} \signal{x}_n=\signal{x}\Rightarrow \lim_{n\to\infty} f(\signal{x}_n)=f(\signal{x})$, and we write $\lim_{n\to\infty}\signal{x}_n=\signal{x}$ if $\lim_{n\to\infty}\|\signal{x}_n-\signal{x}\|=0$ for some (and hence for every) norm $\|\cdot\|$ in $\real^N$. The notions of upper and lower semicontinuity for proper functions $f:\real^N\to \real \cup \{\infty\}$ restricted to sets $C\subset \mathrm{dom}~f \subset \real^N$ are defined similarly.  Given two functions $f:\real_+\to\real_+$ and $g:\real_+\to\real_+$ we say that $f$ \emph{scales} as $\Theta(g(x))$ as $x\to\infty$ (or, in set notation, $f(x)\in\Theta(g(x))$ as $x\to\infty$) if $(\exists k_0\in\real_{++}) (\exists k_1\in\real_{++}) (\exists x_0\in\real_{++})(\forall x\in\real_+) 
x\ge x_0 \Rightarrow k_0 g(x)\le f(x)\le k_1 g(x)$. If $g$ is a constant function, then we use the convention $f(x)\in\Theta(1)$. The scaling of functions for $x\to 0^+$ is defined similarly. We  consider a positive neighborhood of zero in the above definition.

 Given an arbitrary pair of vectors $(\signal{x},\signal{y})\in\real^N\times\real^N$, the inequality  $\signal{x}\le\signal{y}$ should be understood coordinate-wise. A norm $\|\cdot\|$ in $\real^N$ is said to be \emph{monotone} if $(\forall\signal{x}\in\real^N)(\forall\signal{y}\in\real^N)$ $\signal{0}\le\signal{x}\le\signal{y}\Rightarrow\|\signal{x}\| \le \|\signal{y}\|$, and we refer readers to \cite{john91} for nonequivalent notions of monotonicity that are also common in the literature. By $(\cdot)^t$ we denote the transpose of a vector or matrix.

 For clarity of exposition, we say that $T:C\to D$ is a \emph{mapping} only if $D\subset C \subset \real^N$. In other cases, we use the word \emph{function}. The set of fixed points of a mapping $T:C\to C$, $C\subset\real^N$, is denoted by $\mathrm{Fix}(T):=\{\signal{x}\in C~|~\signal{x}=T(\signal{x})\}$. If $C\subset \real^N$ is a convex set, we say that a mapping $T:C\to\real^N:\signal{x}\mapsto[t_1(\signal{x}), \cdots, t_N(\signal{x})]$ is {\it concave} if the function $t_i:C\to\real$ is concave for every $i\in\{1,\ldots,N\}$. Given a mapping $T:C\to C$, we denote by $T^n$ the $n$-fold composition of $T$ with itself.

The fundamental mathematical tool used in this study is the analytic representation of asymptotic functions, which we state as a definition:

\begin{definition}
	\label{def.asymp_func} {(\cite[Theorem~2.5.1]{aus03} Asymptotic function)}
	The asymptotic function associated with a proper function $f:\real^N\to\real\cup \{\infty\}$ is the function given by 
\begin{align*}
f_{\infty}:\real^N  \to  \real \cup \{-\infty,~\infty\}:\signal{x}\mapsto \liminf_{h\to\infty,\signal{y}\to\signal{x}}{f(h\signal{y})}/{h},
\end{align*}
 or, equivalently,
	\begin{align*}
		\begin{array}{rl}
		f_{\infty}:\real^N  \to & \real \cup \{-\infty,~\infty\} \\ 
		\signal{x}  \mapsto & \inf\left\{\displaystyle \liminf_{n\to \infty} \dfrac{f(h_n\signal{x}_n)}{h_n}~|~h_n\to\infty,~\signal{x}_n\to\signal{x} \right\},
		\end{array}
	\end{align*} 
		where $(\signal{x}_n)_{n\in\Natural}$ and $(h_n)_{n\in\Natural}$ are sequences in $\real^N$ and $\real$, respectively.
\end{definition}

Obtaining asymptotic functions associated with arbitrary proper functions directly from the definition can be difficult, but we show in the next subsection that, for some functions that are common in network feasibility and optimization problems (more precisely, for the functions in Definition~\ref{def.inter_func} below), their associated asymptotic functions are easy to compute [see, in particular,  Proposition~\ref{prop.af_properties}(i)]. 

\subsection{Interference functions}

The main contributions of this study are obtained from fundamental properties of the following classes of interference functions:

\begin{definition}
	\label{def.inter_func}
	Consider the following statements for a (proper) function $f: \real^N \to \real_{+} \cup \{\infty\}$ with $\mathrm{dom} f = \real_{+}^N$: \par 

		\item[{[P1]}] ({\it Scalability}) $(\forall \signal{x}\in\real^N_+)$ $(\forall \alpha>1)$  $\alpha {f}(\signal{x})>f(\alpha\signal{x})$. \par
		
		\item[{[P2]}] ({\it Weak scalability}) $(\forall \signal{x}\in\real^N_+)$ $(\forall \alpha\ge 1)$  $\alpha {f}(\signal{x}) \ge f(\alpha\signal{x})$. \par
				
		\item[{[P3]}] ({\it Homogeneity}) $(\forall \alpha \ge 0)$ $(\forall \signal{x}\in\real_+^N)~ {f}(\alpha \signal{x}) = \alpha f(\signal{x})$. \par		
		
		\item[{[P4]}] ({\it Monotonicity}) $(\forall \signal{x}_1\in\real_+^N)$ $(\forall \signal{x}_2\in\real_+^N)$ $\signal{x}_1\ge\signal{x}_2 \Rightarrow{f}(\signal{x}_1)\ge f(\signal{x}_2)$. \par

If (P1) and (P4) are satisfied, then $f$ is said to be a standard interference (SI) function  \cite{yates95}. If (P2) and (P4) are satisfied, then $f$ is said to be a weakly standard interference (WSI) function. If (P3) and (P4) are satisfied, then $f$ is said to be a general interference (GI) function  \cite[p.~4]{martin11}.~\footnote{In the original definition in \cite[p.~4]{martin11}, general interference functions should also satisfy $f(\signal{x})>0$ for some $\signal{x}\in\real_{++}^N$, but we do not impose this requirement here.}

\end{definition}

We now show some useful relations among the classes of functions in Definition~\ref{def.inter_func} and also among classes of self-mappings constructed with these functions. We start with a lemma that can be easily proved by replacing strict inequalities with weak inequalities in the proof of \cite[Proposition~1]{renato2016}. Details of the proof are omitted for brevity.

\begin{lemma}
	\label{lemma.concave}
	Let $f:\real_+^N\to\real_+ $ be a nonnegative concave (NC) function. Consider the extension given by
	\begin{align}
	\label{eq.ext}
  \tilde{f}:\real^N\to\real_+\cup\{\infty\}: \signal{x}\mapsto 
    \begin{cases}
	  f(\signal{x}), &\text{if }\signal{x}\in \real_+^N \\
	  \infty, &\text{otherwise}.
    \end{cases}
\end{align}
Then $\tilde{f}:\real^N\to\real_+\cup\{\infty\}$ is a WSI function. In particular, if the range of ${f}$ is in the set of positive reals [or, equivalently, $f:\real^N_+\to\real_{++}$ is a positive concave (PC) function], then $\tilde{f}$ is a SI function.
\end{lemma}

The functions in Definition~\ref{def.inter_func} are typically used to construct self-mappings that have many applications in wireless networks  (see, for example 
 \cite{slawomir09,martin11,yates95,nuzman07,renato2016,feh2013}). In particular, in this study we consider the following mappings:

\begin{definition}
	\label{definition.mappings}
	Given $N\in\Natural$ functions $f_i:\real_+^N \to \real_+$, $i\in\{1,\ldots,N\}$, we say that the self-mapping $T:\real^N_+\to \real^N_+:\signal{x}\mapsto[f_1(\signal{x}),\ldots, f_N(\signal{x})]$ is a weakly standard interference mapping (WSI mapping) if all coordinate functions $f_1,\ldots,f_N$ extended as done in \refeq{eq.ext} are WSI functions. Standard interference mappings (SI mappings),  general interference mappings (GI mappings), nonnegative concave mappings (NC mappings), and positive concave mappings (PC mappings) are defined similarly [NOTE: for NC and PC mappings, only the domain  $\real_+^N$ is considered to determine concavity of the extensions in \refeq{lemma.concave}].
\end{definition}

For convenience, denote by $\mathcal{F}_\mathrm{WSI}$, $\mathcal{F}_\mathrm{SI}$, $\mathcal{F}_\mathrm{GI}$, $\mathcal{F}_\mathrm{NC}$, and $\mathcal{F}_\mathrm{PC}$ the sets of, respectively, WSI mappings, SI mappings, GI mappings, NC mappings, and PC mappings. The following relations follow directly from the definitions of the mappings: $\mathcal{F}_\mathrm{SI} \subset \mathcal{F}_\mathrm{WSI}$, $\mathcal{F}_\mathrm{GI} \subset \mathcal{F}_\mathrm{WSI}$, and $\mathcal{F}_\mathrm{PC} \subset \mathcal{F}_\mathrm{NC}$. Furthermore, by Lemma~\ref{lemma.concave}, we have $\mathcal{F}_\mathrm{PC}\subset \mathcal{F}_\mathrm{SI}\subset \mathcal{F}_\mathrm{WSI}$ and $\mathcal{F}_\mathrm{NC}\subset \mathcal{F}_\mathrm{WSI}$. Moreover,  \cite[Lemma~1]{leung2004} [see Fact~\ref{fact.p_si}(i) below] and the homogeneity of general interference functions imply that the sets $\mathcal{F}_\mathrm{SI}$ and $\mathcal{F}_\mathrm{GI}$ are disjoint. For ease of reference, in Fig.~\ref{fig.classes} we related these sets of mappings with a simple diagram, and, for concreteness, in Example~\ref{example.simple_affine} we show that the above mappings are common in problems in wireless networks.

\begin{figure}
	\begin{center}
		\includegraphics[width=0.6\columnwidth]{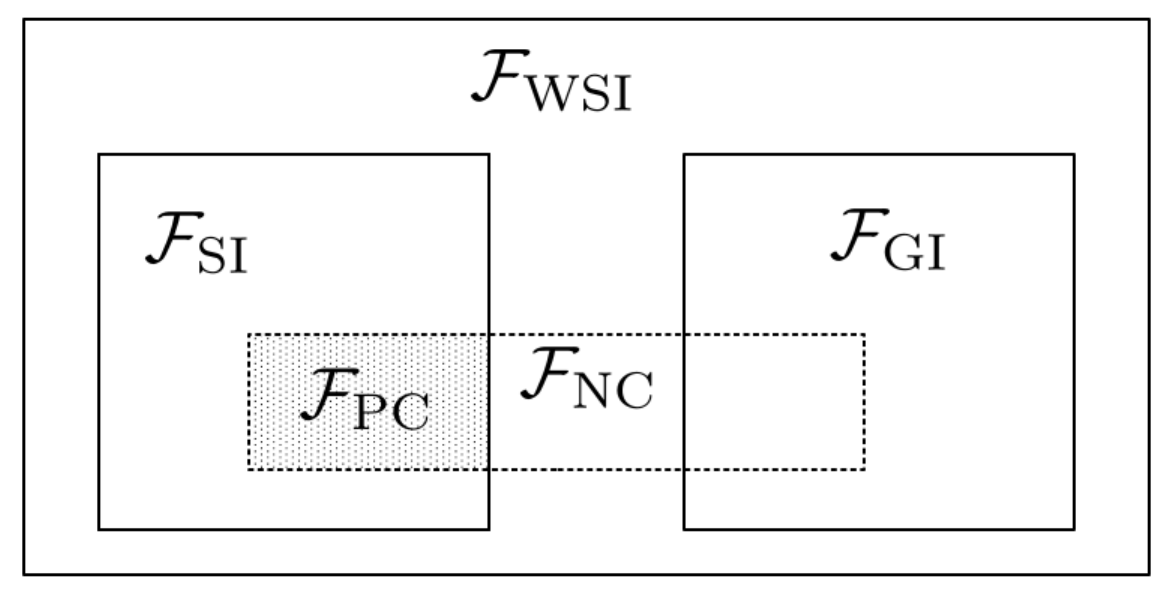}
		\caption{Subclasses of weakly standard interference mappings.} 
		\label{fig.classes}
	\end{center}

\end{figure}

\begin{example}
	\label{example.simple_affine}
	Consider the uplink of a network with one base station and $N$ transmitters represented by the elements of the set $\mathcal{K}=\{1,\ldots, N\}$. Assume that all users share the same wireless resources. Denote by $(p_k)_{k\in\mathcal{K}}\in \real_{++}^N$ and by $(g_k)_{k\in\mathcal{K}}\in \real_{++}^N$, respectively, the power and the pathloss of the transmitters. A widely used performance measure indicating the quality of the wireless links is the signal-to-interference-plus-noise ratio (SINR), which, for the  $k$th user, is given by $\gamma_k:= p_k g_k / (\sum_{j\in\mathcal{K}\backslash\{k\}}p_jg_j+\sigma_k)$, where $\sigma_k\in\real_+$ is the noise power. By rearranging terms of the definition of the SINR, we verify that $ (\forall k\in\mathcal{K})~ p_k = (\gamma_k/g_k) (\sum_{j\in\mathcal{K}\backslash\{k\}}p_jg_j+\sigma_k)$, and we conclude that the power vector $\signal{p}=[p_1,\ldots,p_N]^t$ is a fixed point of the affine mapping $T:\real_+^N\to\real_+^N:\signal{p}\mapsto [t_1(\signal{p}),\ldots, t_N(\signal{p})]^t$, where $(\forall k\in\mathcal{K})~t_k(\signal{p})=(\gamma_k/g_k) (\sum_{j\in\mathcal{K}\backslash\{k\}}p_jg_j+\sigma_k)$. This mapping can be equivalently written as $(\forall \signal{p}\in\real_+^N)~T(\signal{p})=\signal{A}\signal{p}+\signal{u}$, where $\signal{u}=[\sigma_1,\ldots,\sigma_N]^t \in\real_+^K$, and $\signal{A}\in\real_+^{N\times N}$ is a matrix constructed with the SINR values $(\gamma_k)_{k\in\mathcal{K}}$ and the pathloss $(g_k)_{k\in\mathcal{K}}$. Regardless of the values taken by the matrix $\signal{A}\in\real_+^{N\times N}$ and by the vector $\signal{u}\in\real_+^N$, the mapping  $T$ is a member of $\mathcal{F}_\mathrm{NC}$, and hence also of $\mathcal{F}_\mathrm{WSI}$. If $\signal{u}=\signal{0}$, as commonly used in signal-to-interference (SIR) balancing problems \cite[Ch.~6]{martin11}, the mapping $T$ is also a member of $\mathcal{F}_\mathrm{GI}$, but not of $\mathcal{F}_\mathrm{SI}$. If $\signal{u}\in\real_{++}^N$, we have a positive concave mapping (i.e., $T\in \mathcal{F}_\mathrm{PC}$), and the mapping $T$ is a member of $\mathcal{F}_\mathrm{SI}$, but not of $\mathcal{F}_\mathrm{GI}$. For $\signal{u}\in\real_+^N\backslash \real_{++}^N$ with $\signal{u}\ne\signal{0}$, the mapping $T$ is a member of neither $\mathcal{F}_\mathrm{SI}$ nor $\mathcal{F}_\mathrm{GI}$. 
	
\end{example}

Here, particular focus is devoted to mappings in $\mathcal{F}_\mathrm{SI}$, which have some useful properties that we extensively use in the proof of our main contributions: 

\begin{fact}
	\label{fact.p_si}
	(\emph{Selected properties of mappings in $\mathcal{F}_\mathrm{SI}$}) Let $T:\real_+^N\to\real_+^N$ be a mapping in $\mathcal{F}_\mathrm{SI}$. Then each of the following holds:
	\begin{itemize}
		\item[(i)] $(\forall \signal{x}\in\real_+^N)~ T(\signal{x})\in\real_{++}^N$. 
		\item[(ii)] $\mathrm{Fix}(T)$ is either a singleton or the empty set.
		\item[(iii)] If $\signal{x}\in\real_+^N$ satisfies $T(\signal{x})\ge \signal{x}$ (respectively, $T(\signal{x})\le \signal{x}$), then the vector sequence $(\signal{x}_n:=T^n(\signal{x}))_{n\in\Natural}$ is nondecreasing (respectively, nonincreasing) in each coordinate.
		\item[(iv)] Assume that $\mathrm{Fix}(T)\neq\emptyset$. Given an arbitrary vector $\signal{x}\in\real_+^N$, we have $\lim_{n\to\infty}T^n(\signal{x})=\signal{u}\in\mathrm{Fix}(T)$. \\ 
		\item[(v)] $\mathrm{Fix}(T)\neq\emptyset$ if and only if there exists $\signal{x}\in\real^N_+$ satisfying $T(\signal{x})\le\signal{x}$.
	\end{itemize}
\end{fact}

The proof of Fact~\ref{fact.p_si}(i) is found in \cite[Lemma~1]{leung2004}, and the proof of the remaining properties is found in \cite{yates95}. Hereafter, to emphasize Fact~\ref{fact.p_si}(i), we often write $T:\real_+^N\to\real_{++}^N$ (instead of $T:\real_+^N\to\real_{+}^N$) when dealing with mappings in $\mathcal{F}_\mathrm{SI}$.

In the next sections, we derive properties of solutions to problems involving continuous mappings in $\mathcal{F}_\mathrm{SI}$ by, first, associating these mappings to mappings in $\mathcal{F}_\mathrm{GI}$ and  by, second, studying properties of the  mappings in $\mathcal{F}_\mathrm{GI}$. Although the sets $\mathcal{F}_\mathrm{SI}$ and $\mathcal{F}_\mathrm{GI}$ are disjoint, the mappings in these two sets have in common all properties that are valid for mappings in $\mathcal{F}_\mathrm{WSI}$ (see Fig.~\ref{fig.classes}). Therefore, we can avoid duplication of effort to prove these common properties by using mappings in $\mathcal{F}_\mathrm{WSI}$, or their generating coordinate functions, in the statements of our results. For example, the next lemma shows properties that are valid for both $\mathrm{SI}$ and $\mathrm{GI}$ functions. It is a slight extension of results shown in \cite[Lemma~2]{stanczak2009decentralized}, so we omit the proof for brevity.

 \begin{lemma}
 	\label{lemma.basic_properties}
 	Let $f:\real^N \to \real_{+} \cup \{\infty\}$ be a WSI function. Then we have each of the following:

 		\item[(i)] $(\forall \signal{x}\in\real_{+}^N)(\forall \alpha\in~]0,1[) \quad f(\alpha\signal{x}) \ge \alpha f(\signal{x})$, with strict inequality if $f$ is also a SI function.
 		
 		\item[(ii)] $(\forall \signal{x}\in\real_{+}^N)(\forall \alpha_1\in\real_{++})(\forall \alpha_2\in\real_{++}) \\
 		\alpha_2 > \alpha_1 \Rightarrow \dfrac{1}{\alpha_1} f(\alpha_1\signal{x})	\ge \dfrac{1}{\alpha_2} f(\alpha_2\signal{x})$, where the last inequality is strict if $f$ is also a SI function.
 		
 		\item[(iii)] $(\forall \signal{x}\in\real_{+}^N)~ \lim_{h\to\infty} f(h\signal{x})/h\in\real_+$ (i.e., the limit exists for every $\signal{x}\in\real^N_+$).

 \end{lemma}

We now have all ingredients to prove the first main contribution of this study. More precisely, Proposition \ref{prop.af_properties} establishes the connection between  Lemma~\ref{lemma.basic_properties}(iii) and asymptotic functions. This connection is of practical significance because it reveals that a useful simplification for the computation of $f_\infty$ if $f$ is convex \cite[Corollary~2.5.3]{aus03} is also available if $f$ belongs to a class of functions that are not necessarily convex or concave, as common in problems in wireless networks.

\begin{proposition}
	\label{prop.af_properties}
	The asymptotic function $f_{\infty}:\real^N  \to  \real \cup \{\infty\}$ associated with a WSI function $f:\real^N \to \real_{+} \cup \{\infty\}$ has the following properties:

		\item[(i)] $(\forall\signal{x}\in\real^N_+)~f_\infty(\signal{x}) = \lim_{h\to\infty} f(h \signal{x})/{h} \in\real_+$.
		
		\item[(ii)] $f_\infty$ is lower semicontinuous and positively homogeneous. If $f$ is in addition continuous if restricted to the nonnegative orthant $\real_{+}^N$, then $f_\infty$ is continuous if restricted to the nonnegative orthant $\real_{+}^N$.
		\item[(iii)] [{\it Monotonicity}] $(\forall \signal{x}_1\in\real_+^N)$ $(\forall \signal{x}_2\in\real_+^N)$ $\signal{x}_1\ge\signal{x}_2  \Rightarrow  {f}_\infty(\signal{x}_1)\ge f_\infty(\signal{x}_2)$.
		\item[(iv)] Let  $\signal{x}\in\real^N_{+}$ be arbitrary. If $f$ is continuous if restricted to $\real_{+}^N$, then
		$f_\infty(\signal{x})=\lim_{n\to\infty}{f(h_n \signal{x}_n)}/{h_n}$ for all sequences $(\signal{x}_n)_{n\in\Natural}\subset\real_{+}^N$ and $(h_n)_{n\in\Natural}\subset\real_{++}$ such that $\lim_{n\to\infty}\signal{x}_n=\signal{x}$ and $\lim_{n\to\infty}h_n=\infty$.
		\item[(v)] If $f$ is also concave if restricted to $\real_{+}^N$, then $f_\infty$ is concave if restricted to $\real_{+}^N$.

\end{proposition}
\vspace{-.5\baselineskip}
\begin{proof}
	(i) The inequality $(\forall\signal{x}\in\real^N_+)~f_\infty(\signal{x}) \le \lim_{h\to\infty} f(h \signal{x})/{h}\in\real_+$ is immediate from Lemma~\ref{lemma.basic_properties}(iii) and the definition of asymptotic functions, so it is sufficient to prove that $(\forall\signal{x}\in\real^N_+)~f_\infty(\signal{x}) \ge \lim_{h\to\infty} f(h \signal{x})/{h}$ to obtain the desired result.
	
	Let $\alpha\in~]0,1[$ be arbitrary. From definition of asymptotic functions, we know that, for $\signal{x}\in\real_{+}^N$ arbitrary, there exist a sequence $(\signal{x}_n)_{n\in\Natural}\subset \real_{+}^N$ and an  increasing sequence $(h_n)_{n\in\Natural}\subset \real_{++}$ such that $f_\infty(\signal{x})=\lim_{n\to\infty}{f(h_n\signal{x}_n)}/{h_n}$, $\lim_{n\to\infty} \signal{x}_n=\signal{x}$, and $\lim_{n\to\infty} h_n=\infty$. Therefore, as an implication of $\lim_{n\to\infty} \signal{x}_n=\signal{x}\in\real_+^N$, we have $\alpha\signal{x}\le\signal{x}_n$ for every $n\ge L$ with $L\in\Natural$ sufficiently large. Lemma~\ref{lemma.basic_properties}(i) and the monotonicity of WSI functions yield
	$
	(\forall n\ge L)
	~~\alpha{f(h_n\signal{x})}/{h_n}\le { f(\alpha h_n\signal{x})}/{h_n} \le { f(h_n\signal{x}_n)}/{h_n}$.	Taking the limit as $n\to\infty$ and considering Lemma~\ref{lemma.basic_properties}(iii), we verify that 
	\begin{align}
	\label{eq.limit_afunc}
	\alpha \lim_{n\to\infty} \dfrac{f(h_n\signal{x})}{h_n} = \alpha {\lim_{h\to\infty}}\dfrac{f(h\signal{x})}{h} \le f_\infty(\signal{x}).
	\end{align}
	By letting $\alpha\to 1^-$, \refeq{eq.limit_afunc} shows that $f_\infty(\signal{x}) \ge  \lim_{h\to\infty} {f(h\signal{x})}/{h} \ge 0$, which completes the proof.

	(ii) The first part follows directly from \cite[Proposition~2.5.1(a)]{aus03}, which states that asymptotic functions associated with proper functions are lower semicontinuous and positively homogeneous. To prove the second part, assume that $f$ is continuous if restricted to $\real_{+}^N$, and let $(h_n)_{n\in\Natural}\subset\real_{++}$ be an arbitrary increasing sequence such that $\lim_{n\to\infty}h_n=\infty$. For each $n\in\Natural$, define $g_n:\real_+^N\to\real_+:\signal{x}\mapsto {f(h_n \signal{x})}/{h_n}$, and note that continuity of $f$ if restricted to $\real_+^N$  implies that $g_n$ is also continuous.  The property in (i) and Lemma~\ref{lemma.basic_properties}(ii)-(iii) imply that $(\forall \signal{x}\in\real_{+}^N) f_\infty(\signal{x})=\inf_{n\in\Natural} g_n(\signal{x})$. This shows that, if restricted to $\real_{+}^N$, $f_\infty$ is the pointwise infimum of continuous functions, so  $f_\infty$ is upper semicontinuous (see~\cite[Lemma~1.26]{baus11}), in addition to being lower semicontinuous as already shown. Therefore, $f_\infty$ restricted to $\real_{+}^N$ is continuous.

	(iii) Let $({h}_n)_{n\in\Natural}\subset\real_{++}$ be an arbitrary monotone sequence such that $\lim_{n\to\infty}h_n=\infty$. If $\signal{x}_1\ge\signal{x}_2\ge\signal{0}$, the monotonicity property of WSI functions shows that 
	$
	(\forall n\in\Natural)~
	{f(h_n\signal{x}_2)}/{h_n} \le  {f(h_n\signal{x}_1)}/{h_n}.$ Now let $n\to\infty$ and use the property in (i) to obtain the desired result
	$
	f_\infty(\signal{x}_2) \le  f_\infty(\signal{x}_1).
	$
	
	(iv) Let the arbitrary sequences $(\signal{x}_n)_{n\in\Natural}\subset\real_{+}^N$ and $(h_n)_{n\in\Natural}\subset\real_{++}$ satisfy $\lim_{n\to\infty}\signal{x}_n=\signal{x}$ and $\lim_{n\to\infty}h_n=\infty$. Denote by $\signal{1}\in\real^N$ the vector of ones. As a consequence of $\lim_{n\to\infty}\signal{x}_n=\signal{x}$, we know that 
	\begin{align}
	\label{eq.ineqxn}
	(\forall \epsilon>0)(\exists L\in\Natural)(\forall n\ge L)~ \signal{x}_n\le\signal{x}+\epsilon\signal{1}.
	\end{align}
	As a result,  for every $\epsilon>0$, we have
	\begin{multline}
	\label{eq.ineq_2ef}
	f_\infty(\signal{x}) \overset{(a)}{\le} \liminf_{n\to\infty} \dfrac{1}{{h}_n} f({h}_n\signal{x}_n) {\le} \limsup_{n\to\infty} \dfrac{1}{{h}_n} f({h}_n\signal{x}_n) \\ \overset{(b)}{\le} \limsup_{n\to\infty} \dfrac{1}{{h}_n} f({h}_n(\signal{x}+\epsilon\signal{1})) \overset{(c)}{=} f_\infty(\signal{x}+\epsilon\signal{1}),
	\end{multline}
	where (a) follows from the definition of asymptotic functions, (b) follows from \refeq{eq.ineqxn} and  monotonicity of $f$, and (c) is a consequence of the property we proved in (i).  By assumption, $f$ restricted to $\real_{+}^N$ is continuous, so $f_\infty$ restricted to $\real_{+}^N$ is also continuous  as shown in (ii). Therefore, by \refeq{eq.ineq_2ef} and continuity of $f_\infty$ if restricted to $\real_+^N$, we have
	\begin{multline*}
	f_\infty(\signal{x}) {\le} \liminf_{n\to\infty} \dfrac{1}{{h}_n} f({h}_n\signal{x}_n) {\le} \limsup_{n\to\infty} \dfrac{1}{{h}_n} f({h}_n\signal{x}_n) \\ {\le} \lim_{\epsilon\to 0^+} f_\infty(\signal{x}+\epsilon\signal{1})=f_\infty(\signal{x}),
	\end{multline*}	
	which implies that $\lim_{n\to\infty} f({h}_n\signal{x}_n)/{{h}_n}=f_\infty(\signal{x})$.
	
	(v) The proof follows directly from the definition of concave functions and (i).\end{proof}

The results in the previous proposition enable us to introduce a new mapping (or a transform) that plays a key role in this study:

\begin{definition}
	\label{def.amap} (Asymptotic mappings)
	 Let $T: \real^N_+ \to\real^N_{+}:\signal{x}\mapsto [t^{(1)}(\signal{x}),\cdots,t^{(N)}(\signal{x})]$ be a mapping in $\mathcal{F}_\mathrm{WSI}$, where $t^{(i)}:\real^N_+\to\real_{+}$ for each $i\in\{1,\ldots,N\}$. The asymptotic mapping $T_\infty:\real_+^N\to\real_+^N$ associated with $T\in\mathcal{F}_\mathrm{WSI}$ is defined to be the mapping given by  
	$
	T_\infty:\real^N_+\to\real^N_+:\signal{x}\mapsto [\tilde{t}^{(1)}_\infty(\signal{x}),\cdots,\tilde{t}^{(N)}_\infty(\signal{x})],
	$
	where, for each $i\in\{1,\cdots,N\}$, $\tilde{t}^{(i)}_\infty$ is the asymptotic function associated with \linebreak $\tilde{t}^{(i)}:\real^N\to\real_+\cup\{\infty\}:\signal{x} \mapsto \begin{cases} t^{(i)}(\signal{x}),&\text{if }\signal{x}\in\real^N_+ \\
	\infty&\text{otherwise}
	\end{cases}$. 
\end{definition}

As shown in Remark~\ref{remark.example_asym} and Example~\ref{example.lhopital} below, asymptotic mappings can be conveniently obtained by using Proposition~\ref{prop.af_properties}(i), possibly combined with the standard L'H\^opital's rule:

\begin{remark}
	\label{remark.example_asym}
	 Let $T:\real_+^N\to\real_+^N$ be an arbitrary mapping in $\mathcal{F}_\mathrm{WSI}$. Then, by Proposition~\ref{prop.af_properties}(i), we have 
	\begin{align}
	\label{eq.simple_computation}
	(\forall \signal{x}\in\real_+^N)~ T_\infty(\signal{x})=\lim_{h\to\infty} ((1/h)~T(h\signal{x})).
	\end{align}
\end{remark}

\begin{example} 
	\label{example.lhopital}
	Let $\alpha\in\real_+$ be fixed and define $f_\alpha:\real^2_+\to\real_{++}:(x_1,x_2)\mapsto \ln (1+x_2)+\alpha x_1+0.1$ and $g:\real_+^2\to\real_{++}:(x_1,x_2)\mapsto \sqrt{x_1+x_2+1}$.  Consider the mapping given by $(\forall \signal{x}\in\real_+^2)~T_\alpha(\signal{x})= [f_\alpha(x_1,x_2),~ g(x_1, x_2)]^t$, where $[x_1,x_2]^t:=\signal{x}$. We can verify that $T_\alpha$ is a positive concave mapping, so it is a member of $\mathcal{F}_\mathrm{SI}$ (see Fig.~\ref{fig.classes}).  Using Remark~\ref{remark.example_asym} and L'H\^opital's rule, we deduce $(\forall\signal{x}\in\real^2_+)~(T_\alpha)_\infty(\signal{x})=\lim_{h\to\infty}~[f_\alpha(h x_1, h x_2)/h, g(hx_1, h x_2)/h]^t = \lim_{h\to\infty}\left[\dfrac{\mathrm{d}}{\mathrm{d}h}f_\alpha(hx_1,hx_2),~\dfrac{\mathrm{d}}{\mathrm{d}h} g(hx_1, h x_2)\right]^t = \signal{Ax}$, where
$\signal{A}=\left[\begin{matrix}\alpha & 0 \\  0 & 0\end{matrix} \right]$.
\end{example}

 The study in \cite{boche2010} (see also \cite[Theorem~2.14]{martin11}) has shown that standard interference mappings and general interference mappings are deeply connected. The next corollary, an immediate application of Proposition~\ref{prop.af_properties}(i)-(iii), establishes an alternative view of this connection, this time using the concept of asymptotic mappings in Definition~\ref{def.amap}. In particular, Corollary~\ref{cor.connection}(iii) shows that general interference mappings can be seen as asymptotic mappings associated with standard interference mappings.
 
\begin{Cor}
	\label{cor.connection}

	(i) Mappings in $\mathcal{F}_\mathrm{GI}$ and their asymptotic mappings are the same; i.e., 
	\begin{align*}
	(\forall T\in\mathcal{F}_\mathrm{GI})(\forall \signal{x}\in\real_{+}^N)~T_\infty(\signal{x}) = T(\signal{x}).
	\end{align*}
	(ii) Denote by $\mathcal{F}_\infty:=\{T_\infty~|~T\in\mathcal{F}_\mathrm{WSI}\}$ the set of all asymptotic mappings associated with WSI mappings. Then $\mathcal{F}_\infty=\mathcal{F}_\mathrm{GI}$.	\\
	(iii) Every mapping $G:\real^N_+ \to \real^N_+$ in $\mathcal{F}_\mathrm{GI}$ is the asymptotic mapping $T_\infty:\real^N_+ \to \real^N_{+}$ associated with a (nonunique) mapping $T:\real^N_+ \to \real^N_{++}$ in $\mathcal{F}_\mathrm{SI}$; i.e., 
	\begin{align*} (\forall G\in\mathcal{F}_{\mathrm{GI}}) (\exists T\in\mathcal{F}_{\mathrm{SI}}) (\forall \signal{x}\in\real_+^N)~ G(\signal{x})=T_\infty(\signal{x}).
	\end{align*}
	In particular, let $G\in\mathcal{F}_{\mathrm{GI}}$ and $\signal{u}\in\real_{++}^N$ be arbitrary. Then $G$ is the asymptotic mapping associated with $T:\real_+^N\to\real_{++}^N:\signal{x}\mapsto G(\signal{x})+\signal{u}$. The mapping $T$ constructed in this way is a member of $\mathcal{F}_\mathrm{SI}$, and $T$ is continuous if $G$ is continuous.

\end{Cor}

\subsection{Eigenvalues, eigenvectors, and spectral radius of possibly nonlinear mappings}

As shown in \cite{nuzman07,sun2014,renatomaxmin,tan2014wireless,renato2016maxmin, cai2012optimal,zheng2016,yang1998optimal}, among other studies, some optimization and feasibility problems in wireless networks reduce to computing eigenvectors and eigenvalues of possibly nonlinear mappings $T:\real_+^N\to\real_+^N$ such as those in Definition~\ref{definition.mappings}, and we say that $\signal{x}\in\real^N_+\backslash\{\signal{0}\}$ is an eigenvector associated with an eigenvalue $\lambda\in\real_{+}$ if $T(\signal{x})=\lambda \signal{x}$.  Therefore, we should expect the behavior of many wireless transmission schemes to be deeply connected with spectral properties of the mappings in Definition~\ref{definition.mappings}. In this subsection, we lay the foundations to establish novel connections of this type. We start by recalling the following known result:

\begin{fact}
	\label{fact.nuzman} 
	(\cite{nuzman07,krause01}) Let $\|\cdot\|$ be a monotone norm. Assume that $T:\real_{+}^N\to\real_{+}^N$ satisfies at least one of the properties below:
		\item[(a)] $T\in\mathcal{F}_\mathrm{SI}$;
		\item[(b)] $T\in\mathcal{F}_\mathrm{NC}$ and, for every $\signal{x}\in\real_{+}^N\backslash\{\signal{0}\}$, there exists $m\in\Natural$ such that  $T^n(\signal{x})>\signal{0}$ for all $n\ge m$ (these mappings are said to be \emph{primitive}).

	Then each of the following holds:
		\item[(i)] There exists a unique solution $(\signal{x}^\star, \lambda^\star)\in\real_{++}^N\times\real_{++}$ to the conditional eigenvalue problem
		\begin{problem}
			\label{problem.cond_eig}
			Find $(\signal{x}, \lambda)\in\real_{+}^N\times\real_{+}$ such that $T(\signal{x})=\lambda\signal{x}$ and $\|\signal{x}\|=1$.
		\end{problem}

		\item[(ii)] The sequence $(\signal{x}_n)_{n\in\Natural}\subset\real_{+}^N$ generated by 
		\begin{align}
		\label{eq.krause_iter}
		\signal{x}_{n+1} = T^\prime({\signal{x}_n}):=\dfrac{1}{\|T(\signal{x}_n)\|}T(\signal{x}_n),\quad\signal{x}_1\in\real_{++}^N,
		\end{align}
		converges to the uniquely existing vector $\signal{x}^\star\in\mathrm{Fix}(T^\prime):=\{\signal{x}\in\real_{+}^N~|~\signal{x}=T^\prime(\signal{x})\}$, which is also the vector $\signal{x}^\star$ of the tuple $(\signal{x}^\star,\lambda^\star)$ that solves Problem~\ref{problem.cond_eig}. Furthermore, the sequence $(\lambda_n:=\|T(\signal{x}_n)\|)_{n\in\Natural}\subset\real_{++}$ converges to $\lambda^\star$.
\end{fact}

To bound conditional eigenvalues of mappings in $\mathcal{F}_\mathrm{SI}$, we can use the next lemma.

\begin{lemma}
	\label{lemma.lambda_ineq}
	Let $T:\real_+^N\to\real_{++}^N$ be a mapping in $\mathcal{F}_\mathrm{SI}$ and $\|\cdot\|$ an arbitrary monotone norm. If $(\signal{x}^\prime,\lambda^\prime)\in\real_+^N\times\real_+$ satisfies $T(\signal{x}^\prime)\ge \lambda^\prime\signal{x}^\prime$ and $\|\signal{x}^\prime\|=1$, then the scalar $\lambda^\star$ of the solution $(\signal{x}^\star, \lambda^\star)\in\real_{++}^N\times\real_{++}$  to Problem~\ref{problem.cond_eig} satisfies $\lambda^\star \ge \lambda^\prime$.\footnote{See \cite[Lemma~3.2]{nussbaum1986convexity} for a related result concerning mappings in  $\mathcal{F}_\mathrm{GI}$.}
\end{lemma}

\begin{proof}
	 Since $T\in\mathcal{F}_\mathrm{SI}$, Fact~\ref{fact.nuzman} implies that the solution $(\signal{x}^\star,\lambda^\star)\in\real_{++}^N\times\real_{++}$ to Problem~\ref{problem.cond_eig} exists, and $\signal{x}^\star$ is a fixed point of the mapping
	\begin{align*}
	T_{\lambda^\star}:\real_+^N\to\real_{++}^N:\signal{x}\mapsto \dfrac{1}{\lambda^\star} T(\signal{x}).
	\end{align*}
	 Furthermore, we also verify uniqueness of the fixed point (i.e., $\mathrm{Fix}(T_{\lambda^\star})=\{\signal{x}^\star\}$) by recalling that mappings in $\mathcal{F}_\mathrm{SI}$ are closed under multiplication by positive scalars and that these mappings have at most one fixed point [Fact~\ref{fact.p_si}(ii)]. 
Now we prove the lemma by contradiction. Assume that $\lambda^\prime>\lambda^\star$. Then, by $T(\signal{x}^\prime)\ge \lambda^\prime\signal{x}^\prime$, we have
	\begin{align*}
	T_{\lambda^\star}(\signal{x}^\prime)=\dfrac{1}{\lambda^\star} T(\signal{x}^\prime)\ge \dfrac{\lambda^\prime}{\lambda^\star}\signal{x}^\prime \ge \signal{x}^\prime,
	\end{align*}  
	which, together with Fact~\ref{fact.p_si}(iii)-(iv), implies $\signal{x}^\prime\le\signal{x}^\star\in \mathrm{Fix}(T_{\lambda^\star})$. Monotonicity of $T$ shows that $\lambda^\prime\signal{x}^\prime \le T(\signal{x}^\prime)\le T(\signal{x}^\star)=\lambda^\star \signal{x}^\star$, and, from the monotonicity of the norm $\|\cdot\|$, we obtain $\lambda^\prime=\lambda^\prime \|\signal{x}^\prime\| \le \lambda^\star \|\signal{x}^\star\| = \lambda^\star$, which contradicts $\lambda^\prime>\lambda^\star$. As a result, we must have $\lambda^\prime\le \lambda^\star$.
\end{proof}

Many contributions in the next sections are heavily based on the notion of spectral radius of continuous mappings in the set $\mathcal{F}_\mathrm{GI}$, or, equivalently, the set of continuous asymptotic mappings [see Corollary~\ref{cor.connection}(ii)]. The spectral radius of these possibly nonlinear mappings is defined as follows:

\begin{definition}
	\label{def.nl_radius} (\cite[Definition~3.2]{nussbaum1986convexity} Spectral radius) Let $T:\real^N_+\to\real^N_+$ be a continuous mapping in $\mathcal{F}_\mathrm{GI}$. The spectral radius of $T$ is defined to be 
	\begin{align}
	\label{eq.nl_radius}
	\rho(T) := \sup\{\lambda\in\real_+~|~(\exists \signal{x}\in\real_+^N\backslash\{\signal{0}\})~ T(\signal{x})=\lambda\signal{x}\}\in\real_+.
	\end{align}

\end{definition}

It is well known that there always exists a (not necessarily unique) normalized eigenvector associated with the spectral radius of continuous mappings in $\mathcal{F}_\mathrm{GI}$. For later reference, we formally state this result below.

\begin{fact}
	\label{fact.achieved} (\cite[Proposition~5.3.2(ii)]{lem13} and \cite[Corollary~5.4.2]{lem13}) Let $T:\real^N_+\to\real^N_+$ be a continuous mapping in $\mathcal{F}_\mathrm{GI}$. Then, for any norm $\|\cdot\|$, there exists $\signal{x}^\star\in C:=\{\signal{x}\in\real_+^N~|~\|\signal{x}\|=1 \}$ such that $T(\signal{x}^\star)=\rho(T)\signal{x}^\star$.
\end{fact}

For the notion of (nonlinear) spectral radius to be useful in real-world applications, we require efficient means for its computation. To this end, we can use the fixed point iterations in \refeq{eq.krause_iter} as shown in Remark~\ref{remark.comp_specrad} below, which is based on the following known result:

\begin{fact}
\label{fact.nLemma33}
\cite[Lemma~3.3]{nussbaum1986convexity}~Assume that $T\in\mathcal{F}_\mathrm{GI}$ is continuous. Let $(\signal{x}, \lambda)\in\real_{++}^N\times \real_{+}$ satisfy $T(\signal{x})\le \lambda \signal{x}$. Then $\rho(T)\le \lambda$.
\end{fact}

 \begin{remark}
 	\label{remark.comp_specrad}
 	Let $T_\infty:\real_+^N\to\real_{+}^N$ be an asymptotic mapping associated with a continuous mapping $T\in \mathcal{F}_\mathrm{WSI}$. By Proposition~\ref{prop.af_properties}(ii), we know that $T_\infty$ is continuous, and, by Corollary~\ref{cor.connection}, we have $T_\infty\in\mathcal{F}_\mathrm{GI}$. Now assume that the sequence $(\signal{x}_n)\subset\real_+^N$ generated by $(\forall n\in\Natural)~\signal{x}_{n+1}:=(1/\|T_\infty(\signal{x}_n)\|)~T_\infty(\signal{x}_n)$, where $\signal{x}_1\in\real_+^N\backslash\{\signal{0}\}$ and $\|\cdot\|$ is a monotone norm, converges to a positive vector $\signal{x}^\star\in\real_{++}^N$ (see Fact~\ref{fact.nuzman} for \emph{sufficient} conditions). By continuity of $T_\infty$, we deduce $\signal{x}^\star:=T_\infty(\signal{x}^\star)/\|T_\infty(\signal{x}^\star)\|$, and we conclude that $\rho(T_\infty) = \|T_\infty(\signal{x}^\star)\|$ by Fact~\ref{fact.nLemma33} and the definition of the spectral radius.
 \end{remark}

\section{Properties of Some Max-Min Utility Maximization Problems Involving Standard Interference Mappings}
\label{sect.bounds}

We now use the mathematical machinery introduced in the previous section to obtain valuable insights into solutions to the following subclass of utility optimization problems originally studied in \cite{nuzman07}:

\begin{problem}
	\label{problem.canonical}
	(Canonical network utility maximization problem)
	\begin{align}
	\label{eq.canonical}
	\begin{array}{lll}
	\mathrm{maximize}_{\signal{p}, u} & u \\
	\mathrm{subject~to} & \signal{p}\in \mathrm{Fix}(uT):=\left\{\signal{p}\in\real_{+}^N~|~\signal{p}=uT(\signal{p})\right\} \\
	& \|\signal{p}\|_a \le \bar{p} \\
	& \signal{p}\in\real_{+}^N, u\in\real_{++},
	\end{array}
	\end{align}
	where $\bar{p}\in \real_{++}$ is a design parameter hereafter called budget,  $\|\cdot\|_a$ is a monotone norm, and $T:\real_{+}^N\to\real_{++}^N$ is a mapping in $\mathcal{F}_\mathrm{SI}$. 
\end{problem}

Note that well-known (weighted) max-min utility optimization problems reduce to particular instances of \refeq{eq.canonical}, and the studies in \cite{cai2012optimal,tan2014wireless,zheng2016} show simple techniques that can be adapted to identify max-min utility optimization problems that can be posed in the canonical form. Examples of these problems include the (max-min) rate optimization in load coupled networks \cite{renato2016maxmin,siomina2015b}, the joint optimization of the uplink power and the cell assignment \cite{sun2014},  the joint optimization of the uplink power and the receive beamforming vectors \cite[Sect.~1.4.2]{martin11}, and many of the applications described in \cite{cai2012optimal,tan2014wireless,zheng2016}. 

Typically, in network utility maximization problems written in the canonical form \refeq{eq.canonical}, the optimization variable $\signal{p}$ corresponds to the transmit power or load of network elements (e.g., base stations or user equipment); the optimization variable $u$, hereafter called {\it utility}, is the common desired rate or SINR of the users; $T$ is a known  mapping that captures the interference coupling among network elements; and the norm $\|\cdot\|_a$ is chosen based on the physical limitations of network elements. For example, assuming $\signal{p}$ to be a power vector, we can use the $l_1$ norm if all networks elements share the same source of energy, or the $l_\infty$ norm if the network elements have independent sources. For examples of applications where the vector $\signal{p}$ does not have the interpretation of a power vector, we refer the readers to \cite[Sect~V-A]{renato2016maxmin} and \cite{siomina2015b,qi2017}.  

The constraint with the monotone norm in \refeq{problem.canonical} can represent a large class of constraints (on the vector $\signal{p}$) that are common in communication systems. Briefly, any set $C\subset \real_+^N$ with nonempty interior that is downward comprehensible,\footnote{We say that a set $C\subset \real^N_+$ is downward comprehensible if $(\forall \signal{x}\in\real_+^N)(\forall \signal{y}\in C)~\signal{x}\le\signal{y}\Rightarrow \signal{x}\in C$.} compact, and convex can be  written as $C=\{\signal{p}\in\real_{+}^N~|~\|\signal{p}\|_a\le 1\}$ for a monotone norm satisfying $(\forall\signal{p}\in\real_+^N)~\|\signal{p}\|_a=\inf\{t>0~|~(1/t)\signal{p} \in C\}$, and we refer the readers to \cite[Proposition 2]{renatomaxmin} for details.

The objective of the remainder of this section is to derive properties of the solution to Problem~\ref{problem.canonical} as a function of the  budget $\bar{p}$. The main results are formalized in Proposition~\ref{prop.acondv} and Proposition~\ref{prop.bounds}, which show bounds for the utility and the efficiency (i.e., utility over budget -- see Definition~\ref{def.mm_ee}) of the solution to Problem~\ref{problem.canonical}. In particular, these propositions describe the asymptotic behavior of the utility and the efficiency as $\bar{p}\to 0^+$ and as $\bar{p}\to \infty$. In addition, in traditional power control problems in wireless networks, these propositions provide us with a power budget indicating an operation point in which networks are transitioning from a noise-limited regime to an interference-limited regime (see Definition~\ref{def.transition}). To proceed with the formal proof of these contributions, we need the following known fact related to  Problem~\ref{problem.canonical}:

\begin{fact}
	\label{fact.nuzman_main}
	\cite{nuzman07} Denote by $(\signal{p}_{\bar{p}}, u_{\bar{p}})\in\real_{++}^N\times\real_{++}$ a solution to Problem~\ref{problem.canonical} for a given budget $\bar{p}\in\real_{++}$. Then each of the following holds:

		\item[(i)] The solution $(\signal{p}_{\bar{p}}, u_{\bar{p}})\in\real_{++}^N\times\real_{++}$ always exists, and it is unique.
		\item[(ii)] Let $(\signal{v}_{\bar{p}}, \lambda_{\bar{p}})\in\real_{++}^N\times\real_{++}$ be the solution to the following conditional eigenvalue problem: 

		\begin{problem}
			\label{problem.sol_prob3}
			Find $(\signal{v}, \lambda)\in\real_{+}^N\times\real_{+}$ such that $T(\signal{v}) = \lambda \signal{v}$ and $\|\signal{v}\|=1$, where $\|\cdot\|$ denotes the monotone norm $\|\signal{v}\|:=\|\signal{v}\|_a/\bar{p}$.
		\end{problem}		

		Then $\signal{p}_{\bar{p}} = \signal{v}_{\bar{p}}$ and $u_{\bar{p}} = 1/\lambda_{\bar{p}}$.

	\item[(iii)] The function $\real_{++}\to\real_{++}:\bar{p} \mapsto u_{\bar{p}}$ is strictly increasing; i.e., $\bar{p}_1>\bar{p}_2>0$ implies $u_{\bar{p}_1}>u_{\bar{p}_2}>0$.
	\item[(iv)] The function $\real_{++}\to\real_{++}^N:\bar{p} \mapsto \signal{p}_{\bar{p}}$ is strictly increasing in each coordinate; i.e., $\bar{p}_1>\bar{p}_2>0$ implies $\signal{p}_{\bar{p}_1}>\signal{p}_{\bar{p}_2}>\signal{0}$.

\end{fact}

A key implication of Fact~\ref{fact.nuzman_main}(ii) is that the simple fixed point iteration described in Fact~\ref{fact.nuzman}(ii) can be used to solve Problem~\ref{problem.canonical}. Furthermore, Fact~\ref{fact.nuzman_main} shows that the solution to Problem~\ref{problem.canonical} exists and is unique for every $\bar{p}\in\real_{++}$. Therefore, the following functions are well defined:

\begin{definition}
	\label{def.mm_ee}
	(Utility, budget, and $\|\cdot\|_b$-efficiency functions) Denote by  $(\signal{p}_{\bar{p}},~u_{\bar{p}})\in\real_{++}^N\times\real_{++}$ the solution to Problem~\ref{problem.canonical} for a given budget $\bar{p}\in\real_{++}$. The utility and budget functions are defined by, respectively, $U:\real_{++}\to\real_{++}:\bar{p}\mapsto u_{\bar{p}}$ and $P:\real_{++}\to\real_{++}^N:\bar{p}\mapsto \signal{p}_{\bar{p}}$. In turn, given a monotone norm $\|\cdot\|_b$, the $\|\cdot\|_b$-efficiency function is defined by $E:\real_{++}\to\real_{++}:\bar{p}\mapsto U(\bar{p})/\|P(\bar{p})\|_b$, and note that $(\forall \bar{p}\in\real_{++})~ E(\bar{p})=1/\|T(\signal{p}_{\bar{p}})\|_b=u_{\bar{p}}/\|\signal{p}_{\bar{p}}\|_b$.
\end{definition}

To fix the above concepts, we show below an example illustrating that, in some power control problems, the efficiency $E(\bar{p})$ for a given power budget $\bar{p}$ has the intuitive physical interpretation of network-wide \emph{transmit} energy efficiency. However, we emphasize that there are many \emph{nonequivalent} notions of energy efficiency in wireless networks \cite{li2011energy,renato14SPM}. Therefore, any statement related to the behavior of the function $E$ should be put into the right context. For example, power control problems written in the canonical form in \refeq{eq.canonical} often assume a particular transmission scheme (e.g., treat interference as noise). Therefore, statements related to the energy efficiency $E$ for a given power budget $\bar{p}$ do not necessarily agree with, for example, information theoretic notions of energy efficiency, which typically do not impose any restrictions on the transmission scheme or on the available computational power. We also highlight that problems written in the canonical form \refeq{eq.canonical} are not necessarily traditional power control problems \cite[Sect~V-A]{renato2016maxmin}\cite{siomina2015b,qi2017}, so the physical meaning of the efficiency function $E$ can differ depending on the application.

\begin{example} \label{example.intuition} If the norm $\|\cdot\|_b$ is chosen to be the $l_1$ norm in a utility maximization problem in which $\signal{p}$ is a vector of transmit power of base stations (in Watts) and the utility $u$ is the common rate achieved by users (in bits/second) in the downlink, then $E(\bar{p})$ shows the number of bits received by each user for each Joule spent by a network optimized for the power budget $\bar{p}$. In turn, $P(\bar{p})$ and $U(\bar{p})$ represent, respectively, the transmit power of each base station and the optimal rate. In Sect.~\ref{sect.planning} we show a network utility maximization problem with variables having this interpretation.
\end{example}

 By Fact~\ref{fact.nuzman_main}(iii)-(iv),  the utility function $U$ and each coordinate of the budget function $P$ are strictly increasing. However, the next lemma shows that the utility cannot grow faster than the budget.

\begin{lemma} 
	\label{lemma.e_monotone} The $\|\cdot\|_b$-efficiency function $E:\real_{++}\to\real_{++}$ is nonincreasing; i.e., ${\bar{p}_1}>{\bar{p}_2}>0$ implies $E(\bar{p}_1) \le E(\bar{p}_2)$.
\end{lemma}
\vspace{-.5\baselineskip}
\begin{proof}
The result follows from  
\begin{multline*}
 {\bar{p}_1}>{\bar{p}_2}>0 \overset{(a)}{\Rightarrow} P(\bar{p}_1) > P(\bar{p}_2)  \overset{(b)}{\Rightarrow} T(P(\bar{p}_1)) \ge T(P(\bar{p}_2)) \\
 \overset{(c)}{\Rightarrow} \|T(P(\bar{p}_1))\|_b \ge \|T(P(\bar{p}_2))\|_b >0 \overset{(d)}{\Leftrightarrow} E(\bar{p}_1) \le E(\bar{p}_2),
\end{multline*}
where (a) is an implication of Fact~\ref{fact.nuzman_main}(iv), (b) is a consequence of the monotonicity of standard interference functions, (c) results from the monotonicity of the norm $\|\cdot\|_b$ and positivity of $T$, and (d) is obtained from the property $(\forall \bar{p}\in\real_+) E(\bar{p})=1/\|T(P(\bar{p}))\|_b$.
\end{proof}

 (We can prove that the functions $U$, $P$, and $E$ in Definition~\ref{def.mm_ee} are continuous in $\real_{++}$, but we omit the proof for brevity.)  We now have all the technical background to prove the first main result concerning the behavior of the solution to Problem~\ref{problem.canonical} as $\bar{p}\to\infty$:

\begin{proposition}
	\label{prop.acondv}
	Let $(\signal{p}_{\bar{p}}, \lambda_{\bar{p}}):=(P(\bar{p}), 1/U(\bar{p}))\in\real_{++}^N\times\real_{++}$ be the solution to Problem~\ref{problem.sol_prob3} for a given budget $\bar{p}\in\real_{++}$, and denote by $T_\infty:\real_{+}^N\to\real_{+}^N$ the asymptotic mapping associated with the standard interference mapping $T:\real_{+}^N\to\real_{++}^N$. Then each of the following holds:

		\item[(i)] The limit $\lim_{\bar{p}\to\infty}\lambda_{\bar{p}}=:\lambda_\infty\ge 0$ exists.
		\item[(ii)] Let the scalar $\lambda_\infty$ be as defined in (i), and assume that $T$ is continuous. In addition, let $(\bar{p}_n)_{n\in\Natural}\subset \real_{++}$ denote an arbitrary strictly increasing sequence satisfying $\lim_{n\to\infty} \bar{p}_n = \infty$, and define $\signal{x}_{n} := (1/\|\signal{p}_{\bar{p}_n}\|_a)\signal{p}_{\bar{p}_n}=(1/\bar{p}_n)\signal{p}_{\bar{p}_n}$ (recall that $\|\signal{p}_{\bar{p}_n}\|_a=\bar{p}_n$ by Fact~\ref{fact.nuzman_main}(ii)). Let  $\signal{x}_{\infty}\in\real_+^N$ be an arbitrary accumulation point of the bounded sequence $(\signal{x}_n)_{n\in\Natural}\subset\real_{++}^N$. Then the tuple $(\signal{x}_\infty,\lambda_\infty)$ solves the following conditional eigenvalue problem:

		\begin{problem}
			\label{problem.eival_tinf}
			Find $(\signal{x}, \lambda)\in\real_{+}^N\times\real_{+}$ such that $T_\infty(\signal{x}) = \lambda \signal{x}$ and $\|\signal{x}\|_a = 1$.
		\end{problem}	

		\item[(iii)] Furthermore, we have $\lambda_\infty = \rho(T_\infty)$.
		
		\item[(iv)] Moreover, if the solution $(\signal{x}^\prime, \lambda^\prime)\in\real_{+}^N\times\real_{+}$ to Problem~\ref{problem.eival_tinf} is unique, then $\lim_{n\to\infty}\signal{x}_n=\signal{x}_\infty=\signal{x}^\prime$ and $\lambda^\prime=\lambda_\infty$.

\end{proposition}
\vspace{-.5\baselineskip}
\begin{proof}
	 (i) The limit $\lim_{\bar{p}\to\infty} \lambda_{\bar{p}} =: \lambda_\infty\ge 0$ exists because the function $\real_{++}\to\real_{++}:\bar{p} \mapsto \lambda_{\bar{p}}$ is decreasing (and bounded below by zero), as shown in Fact~\ref{fact.nuzman_main}(iii).

	 (ii) Since $(\forall n\in \Natural)~T(\signal{p}_{\bar{p}_n}) = \lambda_{\bar{p}_n} \signal{p}_{\bar{p}_n}$ and $\|\signal{p}_{\bar{p}_n}\|_a = \bar{p}_n>0$, we have
	 \begin{multline}
	 \label{eq.xn}
	 (\forall n\in \Natural) \lambda_{\bar{p}_n} \signal{x}_n = \dfrac{\lambda_{\bar{p}_n}}{\|\signal{p}_{\bar{p}_n}\|_a} \signal{p}_{\bar{p}_n} 
	 = \dfrac{1}{\|\signal{p}_{\bar{p}_n}\|_a} T\left(\dfrac{\|\signal{p}_{\bar{p}_n}\|_a}{\|\signal{p}_{\bar{p}_n}\|_a} \signal{p}_{\bar{p}_n}\right) \\ = \dfrac{1}{\bar{p}_n} T(\bar{p}_n\signal{x}_n).
	 \end{multline}

 The vector $\signal{x}_\infty$ is an accumulation point of the bounded sequence $(\signal{x}_n)_{n\in\Natural}\subset\real_{++}^N$ by definition, so there exists a convergent subsequence $(\signal{x}_n)_{n\in K}$, $K\subset\Natural$, such that $\lim_{n\in K} \signal{x}_{n} = \signal{x}_\infty\in\real_+^N$. Therefore, from part (i) of the proposition and \refeq{eq.xn}, we conclude that
	 \begin{align}
	     \label{eq.limit_tinf}
	 	 \lambda_\infty \signal{x}_\infty = \lim_{n\in K} \dfrac{1}{\bar{p}_n} T(\bar{p}_n\signal{x}_n) = T_\infty(\signal{x}_\infty),
	 \end{align}
	 where the last equality follows from Proposition~\ref{prop.af_properties}(iv) and continuity of $T$. We complete the proof by recalling that  $\|\signal{x}_\infty\|_a = \lim_{n\in K} \|\signal{x}_n\|_a=1$.

(iii) In light of (ii), Definition~\ref{def.nl_radius}, and Fact~\ref{fact.nLemma33}, the result is immediate if there exists a positive accumulation point $\signal{x}_\infty\in\real_{++}^N$. The proof for the general case builds upon the proof of \cite[Theorem~3.1(2)]{nussbaum1986convexity}, which shows a related result for mappings in $\mathcal{F}_\mathrm{GI}$. In more detail, let $\lambda_s := \rho(T_\infty)$, and note that $T_\infty$ is continuous by Proposition~\ref{prop.af_properties}(ii). By Corollary~\ref{cor.connection}(ii) and Fact~\ref{fact.achieved}, there exists $\signal{x}_s\in C:=\{\signal{x}\in\real_+^N~|~\|\signal{x}\|_a=1\}$ such that $T_\infty(\signal{x}_s)=\lambda_s \signal{x}_s$, which, together with Proposition~\ref{prop.af_properties}(i), shows that there exists a sequence $(\delta_n)_{n\in\Natural}\subset \real_+$ satisfying both $\lim_{n\to\infty} \delta_n = 0$ and 
\begin{align}
\label{eq.ineq_deltan}
(\forall n\in\Natural) ~ \dfrac{1}{\bar{p}_n} T(\bar{p}_n\signal{x}_s) \ge (\lambda_s-\delta_n)\signal{x}_s \ge \signal{0}.
\end{align}

By considering the monotone norms $\|\signal{x}\|_{\bar{p}_n}:=(1/\bar{p}_n)\|\signal{x}\|_a$ for every $n\in\Natural$ and every $\signal{x}\in\real_+^N$ (which, in particular, implies $\|\bar{p}_n\signal{x}_s\|_{\bar{p}_n}=1$ and $\|\signal{p}_{\bar{p}_n}\|_{\bar{p}_n}=1$), we have from the inequality in \refeq{eq.ineq_deltan},  Lemma~\ref{lemma.lambda_ineq}, and Fact~\ref{fact.nuzman_main}(ii) that $(\forall n\in\Natural)~\lambda_{\bar{p}_n} \ge (\lambda_s-\delta_n)$. Passing to the limit as $n\to\infty$, we obtain $\lambda_\infty \ge \lambda_s$. In addition, the definition of the spectral radius and the equality $\lambda_\infty\signal{x}_\infty = T_\infty(\signal{x}_\infty)$ in part (ii) show that $\lambda_\infty \le \lambda_s$. Therefore, we need to have $\lambda_\infty = \lambda_s$ as claimed.

(iv) If $(\signal{x}^\prime, \lambda^\prime)\in\real_{++}^N\times\real_{++}$ is the unique solution to Problem~\ref{problem.eival_tinf}, then, as an immediate consequence of the result in (ii), $\signal{x}^\prime$ is the only accumulation point of the bounded sequence $(\signal{x}_n)_{n\in\Natural}$, which implies that $\lim_{n\to\infty}\signal{x}_n=\signal{x}_\infty=\signal{x}^\prime$. As a result, $T_\infty(\signal{x}_\infty) =\lambda^\prime \signal{x}_\infty\in\real_{++}^N$ and $\lambda^\prime=\lambda_\infty$ by also considering the result in (ii).\end{proof}

The assumption of uniqueness of the solution to Problem~\ref{problem.eival_tinf} in Proposition~\ref{prop.acondv} is often valid in utility maximization problems. In particular, Fact~\ref{fact.nuzman} shows sufficient conditions that are easily verifiable. If they are satisfied, then the spectral radius of the asymptotic mapping $T_\infty$ can be obtained directly by solving Problem~\ref{problem.eival_tinf} with the approach in Remark~\ref{remark.comp_specrad}. However, even if the assumptions in Fact~\ref{fact.nuzman} are not satisfied, we discuss below that the iterations in \refeq{eq.krause_iter} can also be used to approximate the spectral radius of arbitrary asymptotic mappings.

\begin{remark}
	\label{remark.approx}
	
	The results in Proposition~\ref{prop.acondv}(ii)-(iii) are valid even if Problem~\ref{problem.eival_tinf} does not have a unique solution. Therefore, if we need to compute the spectral radius of an asymptotic mapping $T_\infty:\real_+^N\to\real_+^N$ associated with a continuous mapping $T\in\mathcal{F}_\mathrm{SI}$, and the assumptions in Fact~\ref{fact.nuzman} are not satisfied for $T_\infty$, then we can approximate $\rho(T_\infty)$ by using $T$ and a sufficiently large budget $\bar{p}$ in Problem~\ref{problem.canonical}. The reciprocal of the utility obtained by solving this problem serves as an estimate of $\rho(T_\infty)$, as justified by  Proposition~\ref{prop.acondv}(i). In fact, this estimate is  an upper bound of $\rho(T_\infty)$ in light of Fact~\ref{fact.nuzman_main}(ii)-(iii), and recall that Problem~\ref{problem.canonical} can always be solved with the algorithm in Fact~\ref{fact.nuzman}(ii).
	 
\end{remark}

With knowledge of the spectral radius of the asymptotic mapping associated with the mapping $T$ in Problem~\ref{problem.canonical}, we show in the next proposition that performance bounds on the  utility $U$ and the efficiency $E$ become readily available. These bounds are asymptotically tight, and they are simple functions of the budget $\bar{p}$. Therefore,  they are particularly useful to gain information on the  performance of wireless networks for a given budget $\bar{p}$ by using operations that typically have remarkably low computational complexity.

\begin{proposition}
\label{prop.bounds}
For notational simplicity, let $\lambda_\infty:=\rho(T_\infty)$ be the spectral radius of the asymptotic mapping $T_\infty$ associated with the mapping $T$ in Problem~\ref{problem.canonical}. Assume that $T$ is continuous and that $\lambda_\infty>0$. Then each of the following holds:

		\item[(i)] $\sup_{\bar{p}>0} U(\bar{p}) = \lim_{\bar{p}\to \infty} U(\bar{p}) = 1/\lambda_\infty$ and $\sup_{\bar{p}>0} E(\bar{p}) = \lim_{\bar{p}\to 0^+} E(\bar{p}) = 1/\|T(\signal{0})\|_b$.
	\item[(ii)] $(\forall \bar{p}\in\real_{++})$ 
	\begin{align*} U(\bar{p}) 
	\le  \begin{cases} \bar{p}/\|T(\signal{0})\|_a,&\text{if } \bar{p}\le\|T(\signal{0})\|_a/\lambda_\infty\\ 1/\lambda_\infty,&\text{otherwise.}\end{cases}
	\end{align*}	
	\item[(iii)] $(\forall \bar{p}\in\real_{++}) E(\bar{p}) \le \min\{1/\|T(\signal{0})\|_b, \alpha/(\lambda_\infty~\bar{p})\}$, where $\alpha\in\real_{++}$ is any scalar satisfying $(\forall\signal{x}\in\real^N)~\|\signal{x}\|_a \le \alpha \|\signal{x}\|_b$ (such a scalar always exists because of the equivalence of norms in finite dimensional spaces).
	
	\item[(iv)] $U(\bar{p})\in\Theta(1)$ and $E(\bar{p})\in\Theta(1/\bar{p})$ as $\bar{p}\to\infty$.
	\item[(v)] $U(\bar{p})\in\Theta(\bar{p})$ and $E(\bar{p})\in\Theta(1)$ as $\bar{p}\to 0^+$.

\end{proposition}

\begin{proof}
	(i) The utility function $U$ is strictly increasing by Fact~\ref{fact.nuzman_main}(iii) and $\lim_{\bar{p}\to\infty} U(\bar{p})= 1/\lambda_\infty$ by Proposition~\ref{prop.acondv}(iii),  so $\sup_{\bar{p}>0} U(\bar{p}) = \lim_{\bar{p}\to \infty} U(\bar{p}) = 1/\lambda_\infty$. 

Now, let $(\bar{p}_n)_{n\in\Natural}\subset \real_{++}$ be an arbitrary sequence such that $\lim_{n\to\infty}\bar{p}_n = 0$. To prove that $\lim_{\bar{p}\to 0^+} E(\bar{p}) = \lim_{\bar{p}\to 0^+} {1}/{\|T(P(\bar{p}))\|_b} = {1}/{\|T(\signal{0})\|_b}$, we only have to show that $\lim_{n\to\infty} {1}/{\|T(P(\bar{p}_n))\|_b} = {1}/{\|T(\signal{0})\|_b}$. By Fact~\ref{fact.nuzman_main}(ii), we have $\lim_{n\to\infty} \|P(\bar{p}_n)\|_a = \lim_{n\to\infty} \bar{p}_n = {0}$, and thus $\lim_{n\to\infty} P(\bar{p}_n)=\signal{0}$. Therefore, by continuity and positivity of $T$, we obtain $\lim_{n\to\infty} E(\bar{p}_n) = \lim_{n\to\infty} {1}/{\|T(P(\bar{p}_n))\|_b}= {1}/{\|T(\signal{0})\|_b}<\infty$. 	The equality $\sup_{\bar{p}>0} E(\bar{p}) = \lim_{\bar{p}\to 0^+} E(\bar{p}) = {1}/{\|T(\signal{0})\|_b}$ is now immediate from Lemma~\ref{lemma.e_monotone}, and the proof of (i) is complete.

(ii) By Fact~\ref{fact.nuzman_main}(ii), positivity of $U$ and $T$, and monotonicity of $T$, $\|\cdot\|_a$, and $P$, we have 
\begin{align*}
(\forall\bar{p}>0)~ 0< U(\bar{p}) \|T(\signal{0})\|_a \le  U(\bar{p}) \|T(P(\bar{p}))\|_a   = \bar{p},
\end{align*}
 and thus $(\forall\bar{p}>0)~U(\bar{p})\le \bar{p}/ \|T(\signal{0})\|_a$. Furthermore, by (i), we also have $(\forall\bar{p}>0)~U(\bar{p})\le 1/\lambda_\infty$. Combining these last two inequalities, we obtain the desired result $(\forall \bar{p}\in\real_{++}) U(\bar{p}) \le  \min\{\bar{p}/\|T(\signal{0})\|_a, 1/\lambda_\infty\}$.
 
 (iii)  By (i), Fact~\ref{fact.nuzman_main}(ii), and the definition of the efficiency function, we deduce for every $\bar{p}>0$:
 \begin{align*}
  E(\bar{p})=\dfrac{U(\bar{p})}{\|P(\bar{p}))\|_b}\le \dfrac{\alpha  U(\bar{p})}{\|P(\bar{p}))\|_a}  = \dfrac{\alpha U(\bar{p})}{\bar{p}} \le \dfrac{\alpha}{\lambda_\infty \bar{p}}. 
 \end{align*}
The desired result is obtained by combining the previous bound with the bound $(\forall \bar{p}>0)~E(\bar{p}) \le 1/\|T(\signal{0})\|_b$, which is immediate from (i).

(iv) The relation  $U(\bar{p})\in\Theta(1)$ is immediate from $\lim_{\bar{p}\to \infty} U(\bar{p}) = 1/\lambda_\infty$, as shown in (i). To prove that $E(\bar{p})\in\Theta(1/\bar{p})$ as $\bar{p}\to\infty$, recall that, from the equivalence of norms in finite dimensional spaces, there exists a scalar $\beta\in\real_{++}$ such that $(\forall \bar{p}>0)~\|P(\bar{p})\|_\infty \le \beta \|P(\bar{p})\|_a=\beta\bar{p}$, which  implies $(\forall \bar{p}>0)~P(\bar{p}) \le \beta\bar{p}\signal{1}$. Now we use the bound in (iii), the monotonicity of the norm $\|\cdot\|_a$, and the monotonicity and scalability properties of standard interference functions to verify that  $(\forall \bar{p}>1)~{\alpha}/{(\lambda_\infty \bar{p})} \ge {E}(\bar{p}) = {1}/{\|T(P(\bar{p}))\|_b} \ge {1}/{
	\|T(\beta \bar{p} \signal{1})\|_b} \ge {1}/{(\bar{p} \|T(\beta  \signal{1})\|_b)}$, which implies $E(\bar{p})\in\Theta(1/\bar{p})$ as $\bar{p}\to\infty$.

(v) The result $E(\bar{p})\in\Theta(1)$ as $\bar{p}\to 0^+$ is immediate from (i). To prove that $U(\bar{p})\in\Theta(\bar{p})$ as $\bar{p}\to 0^+$, denote by $E_a$ the $\|\cdot\|_a$-efficiency (NOTE: the monotone norm used for the efficiency functions $E$ and $E_a$ may differ). By definition of $E_a$ and Fact~\ref{fact.nuzman}(ii), we have $(\forall \bar{p}\in\real_{++})~\bar{p} E_a(\bar{p}) = U(\bar{p})$. Since $E_a$ is nonincreasing as shown in Lemma~\ref{lemma.e_monotone}, for any $\bar{p}_0\in\real_{++}$, we deduce:
\begin{align*}
(\forall \bar{p}\le p_0)~ \bar{p} E_a(\bar{p}_0) \le U(\bar{p}) \le  \bar{p} \lim_{x\to 0^+} E_a(x) = \dfrac{\bar{p}}{\|T(\signal{0})\|_a},
\end{align*}
where in the last equality we use part (i) of the proposition. The above implies $U(\bar{p})\in\Theta(\bar{p})$ as $\bar{p}\to 0^+$, and the proof is complete.
\end{proof}

The bound in Proposition~\ref{prop.bounds}(ii) motivates the definition of the following transition point:

 \begin{definition}
 	\label{def.transition}
 	If the assumptions in Proposition~\ref{prop.bounds} are valid, we say that the network operates in the \emph{low budget regime} if $\bar{p} \le \bar{p}_\mathrm{T}$ or in the \emph{high budget regime} if $\bar{p} > \bar{p}_\mathrm{T}$, where the budget $\bar{p}_\mathrm{T}:=\|T(\signal{0})\|_a/\lambda_\infty$ is the \emph{transition point}.  
 \end{definition}

To explain in simple terms the importance of Proposition~\ref{prop.bounds} and  Definition~\ref{def.transition}, let us consider a network utility optimization problem in which the budget $\bar{p}$ has the unit of power and the utility $U$ is given in bits/second, as shown in Example~\ref{example.intuition}. In this special case, the transition point gives a coarse indication of an operating point in which the performance of networks are transitioning from a noise-limited regime  to an interference-limited regime. In the low power (noise-limited) regime, if the budget $\bar{p}$ is kept sufficiently small, as an implication of Proposition~\ref{prop.bounds}(v), we verify that the budget efficiency $E$ (i.e., the network transmit energy efficiency) decreases slowly with increasing $\bar{p}$. In contrast, with increasing $\bar{p}$, the rate or utility $U$ behaves similarly to a linear function [see Proposition~\ref{prop.bounds}(ii) and Proposition~\ref{prop.bounds}(v)]. The above observations show that we should transmit with low power and low rate if high transmit energy efficiency is desired, and we emphasize that we have proved this expected result by using a very general model that is known to unify, within a single framework, the behavior of a large array of transmission technologies. In the high power (interference-limited) regime, the energy efficiency eventually decays quickly as the power budget $\bar{p}$ diverges to infinity because $E(\bar{p})\in\Theta(1/\bar{p})$ as $\bar{p}\to\infty$ (i.e., multiplying the power $\bar{p}$ by a factor $\alpha$ roughly divides the transmit energy efficiency by the same factor $\alpha$ when $\bar{p}$ is sufficiently large), while gains in rate eventually become marginal because of the uniform bound $(\forall\bar{p}\in\real_{++})~U(\bar{p})\le 1/\rho(T_\infty)$. An illustration of these properties in a concrete utility optimization problem is shown later in Sect.~\ref{sect.planning}. See, in particular, Figs. \ref{fig.utility} and \ref{fig.energy_efficiency} for a graphical illustration of the properties described above.

\section{Existence of fixed points of standard interference mappings}
\label{sect.existence}
In this section we show that spectral properties of asymptotic mappings provide us with the natural generalization of existing results related to the feasibility of network designs. More specifically, systems designers frequently have to determine whether a network configuration is able to support a given traffic profile. As shown in \cite{yates95, nuzman07,martin11,slawomir09,Majewski2010,siomina12,renato2016,feh2013} and many other studies, answering questions of this type often reduces to determining whether a mapping $T$ in  $\mathcal{F}_\mathrm{SI}$ has a fixed point (see also Sect.~\ref{sect.feasibility_load_model} for exemplary applications). If the mapping is affine, then the following result has been widely used in the context of resource allocation in  wireless networks \cite[Ch.~2]{slawomir09}\cite[Sect.~2.5.2]{chiang2008power}. Let $\signal{X}\in\real_+^{N\times N}$ be a nonnegative matrix and $\signal{u}\in\real_{++}^N$ a positive vector. Then the system $\signal{p}=\signal{u}+\signal{X}\signal{p}$ has a positive solution $\signal{p}=(\signal{I}-\signal{X})^{-1}\signal{u}\in\real_{++}^N$, or, equivalently, $\mathrm{Fix}(T)\neq\emptyset$ for $T:\real^N_+\to\real^N_{++}:\signal{p}\mapsto \signal{X}\signal{p}+\signal{u}$,  if and only if the spectral radius of the matrix $\signal{X}$ is strictly less than one. For this mapping $T$, which is affine in $\real_+^N$, a direct application of the result in Remark~\ref{remark.example_asym} shows that the corresponding asymptotic mapping is $T_\infty:\real^N_+\to\real^N_{+}:\signal{p}\mapsto \signal{X}\signal{p}$, and thus the spectral radius of $T_\infty$ (in the sense of Definition~\ref{def.nl_radius}) and the spectral radius of the matrix $\signal{X}$ (in the conventional sense in linear algebra) coincide. This observation suggests that, in feasibility analysis involving nonlinear mappings, the spectral radius of matrices should be replaced by the spectral radius of asymptotic mappings. Proposition~\ref{proposition.existence} is the formal verification of the validity of this approach {for mappings in $\mathcal{F}_\mathrm{SI}$}.

\begin{proposition}
	\label{proposition.existence}
	Let $T:\real^N_+\to\real_{++}^N$ be a continuous mapping in $\mathcal{F}_\mathrm{SI}$. Then $\mathrm{Fix}(T)\neq \emptyset$ if and only if $\rho(T_\infty)<1$.
\end{proposition}
\begin{proof}
	In the proof below, we assume that $(\bar{p}_n)_{n\in\Natural}\subset~ ]1,\infty[$ is an arbitrary strictly increasing sequence satisfying $\lim_{n\to\infty}\bar{p}_n = \infty$. We also recall that $T_\infty$ is continuous as an implication of its definition and  Proposition~\ref{prop.af_properties}(ii).

	Suppose that $T$ has a fixed point denoted by $\signal{x}^\prime\in\mathrm{Fix}(T)\neq\emptyset$, which is unique by Fact~\ref{fact.p_si}(ii). By the scalability property of $T$, we deduce
	\begin{align*}
	(\forall n\in\Natural)~~ \dfrac{1}{\bar{p}_n}T\left(\bar{p}_n\signal{x}^\prime\right)< \dfrac{\bar{p}_n}{\bar{p}_n} T\left( \signal{x}^\prime\right) = T\left(\signal{x}^\prime\right) = \signal{x}^\prime.
	\end{align*}
	With a direct application of Proposition~\ref{prop.af_properties}(i) and  Lemma~\ref{lemma.basic_properties}(ii), we verify that
	\begin{align*}
	T_\infty(\signal{x}^\prime) = \lim_{n\to\infty} \dfrac{1}{\bar{p}_n}T\left(\bar{p}_n\signal{x}^\prime\right) \le \dfrac{1}{\bar{p}_1}T\left(\bar{p}_1\signal{x}^\prime\right) < \signal{x}^\prime\in\real^N_{++},
	\end{align*}
	and the above implies $\rho(T_\infty)<1$ by Corollary~\ref{cor.connection}(ii), continuity of $T_\infty$, and Fact~\ref{fact.nLemma33}.

	Conversely, assume that $\rho(T_\infty)<1$. Consider Problem~\ref{problem.canonical} with the mapping $T$ and an arbitrary monotone norm $\|\cdot\|_a$, and let $(\signal{p}_{\bar{p}},\lambda_{\bar{p}})\in\real^N_{++}\times\real_{++}$ be as defined in Proposition~\ref{prop.acondv}. By Proposition~\ref{prop.acondv}(iii), we have $\lim_{n\to\infty} \lambda_{\bar{p}_n}=\rho(T_\infty)<1$. Therefore, there exists $M\in\Natural$ such that $\lambda_{\bar{p}_M} \in~]0,1[$. We now obtain from the definition of the tuple $(\signal{p}_{\bar{p}},\lambda_{\bar{p}})$ and Fact~\ref{fact.nuzman_main}(i)-(ii) the relation $\signal{p}_{\bar{p}_M} = ({1}/{\lambda_{\bar{p}_M}}) T(\signal{p}_{\bar{p}_M} ) \ge T(\signal{p}_{\bar{p}_M})$, which implies $\mathrm{Fix}(T)\ne\emptyset$ by Fact~\ref{fact.p_si}(v), and the proof is complete. 	

\end{proof}

To use Proposition~\ref{proposition.existence} in practice, we need simple approaches to compute the spectral radius of $T_\infty$, which can be done, for example, with the techniques in Remark~\ref{remark.comp_specrad} and Remark~\ref{remark.approx}. More interesting, in some applications involving nonlinear mappings $T$, their corresponding  asymptotic mappings $T_\infty$ are linear in $\real_+^N$, so computing $\rho(T_\infty)$ reduces to computing the spectral radius of matrices, which is an operation for which well established algorithms with efficient implementations exist \cite[Ch.~7]{golub}. The application shown later in Sect.~\ref{sect.feasibility_load_model} involves mappings with this desirable property, and below we illustrate with a toy example that we can easily verify the existence of fixed points if the asymptotic mappings are linear.

\begin{example}
	\label{example.existence_ta}
	Fix $\alpha\in\real_+$, and let $T_{\alpha}\in\mathcal{F}_\mathrm{SI}$ be the nonlinear mapping in Example~\ref{example.lhopital}. Recall that its asymptotic mapping is given by $({T}_\alpha)_\infty:\real^2_+\to\real_+^2:\signal{x}\mapsto \signal{A}\signal{x}$, where $A=\left[\begin{matrix}\alpha & 0 \\  0 & 0\end{matrix} \right]$. In this case, the spectral radius of the matrix $\signal{A}$, in the conventional sense in linear algebra, coincides with the spectral radius of the asymptotic mapping $(T_\alpha)_\infty$, in the sense of Definition~\ref{def.nl_radius}. Therefore, we conclude that $\mathrm{Fix}(T_\alpha)\neq\emptyset$ if and only if $\rho(({T}_\alpha)_\infty)=\alpha \in [0,~1[$.
\end{example}

 Proposition~\ref{proposition.existence} is mostly useful if we do not require information about the location of the fixed point. If we also need knowledge of whether the possibly existing fixed point satisfies a constraint written in terms of a monotone norm (e.g., to verify whether SINR requirements can be achieved with a power vector $\signal{p}$ satisfying transmit power constraints), the following result can be used:

\begin{proposition}
	\label{proposition.feasibility_constraints}
	Let $\|\cdot\|$ be a monotone norm and $T:\real_+^N\to\real_{++}^N$ a mapping in $\mathcal{F}_{\mathrm{SI}}$. Denote by $(\signal{x}^\star,\lambda^\star)\in\real^N_{++}\times\real_{++}$ the solution to Problem~\ref{problem.cond_eig} with this mapping and norm. Then $\signal{x}^\prime \in \mathrm{Fix}(T)\neq\emptyset$ and $\|\signal{x}^\prime\| \le 1$ if and only if $\lambda^\star \le 1$.
\end{proposition}

\begin{proof}
	Assume that $\lambda^\star \le 1$. As a result, we have ${T}(\signal{x}^\star)=\lambda^\star \signal{x}^\star \le \signal{x}^\star$, which,  by Fact~\ref{fact.p_si} and monotonicity of the norm $\|\cdot\|$, implies the existence of the unique fixed point $\signal{x}^\prime \in \mathrm{Fix}(T)\neq\emptyset$ satisfying   $\|\signal{x}^\prime\| \le \|\signal{x}^\star\|=1$. Conversely, let the fixed point $\signal{x}^\prime \in \mathrm{Fix}(T)\neq\emptyset$ (which is unique  by Fact~\ref{fact.p_si}(ii)) satisfy $\|\signal{x}^\prime\| \le 1$. In particular, if $\|\signal{x}^\prime\|=1$, we have $\lambda^\star=1$ and $\signal{x}^\star=\signal{x}^\prime$, so we only need to consider the case $\|\signal{x}^\prime\| < 1$. To obtain a contradiction, assume that $\lambda^\star > 1$. Therefore, ${T}(\signal{x}^\star)=\lambda^\star \signal{x}^\star > \signal{x}^\star$. As a result, by Fact~\ref{fact.p_si}(iii)-(iv), we have $\signal{x}^\prime = \lim_{n\to\infty} T^n(\signal{x}^\star) > \signal{x}^\star$, which implies the contradiction $1>\|\signal{x}^\prime\| \ge \|\signal{x}^\star\|=1$ because the norm $\|\cdot\|$ is monotone.
\end{proof}

\section{Exemplary applications}
\label{sect.applications}
To show concrete applications of the theory developed in the previous sections, we consider network feasibility and optimization problems based on the load coupled interference model described in \cite{Majewski2010,renato14SPM,renato2016power,siomina12,ho2014data,ho2015,feh2013,siomina2012load,renato2017load}, among other studies. This model approximates the long-term behavior of modern communication systems (and, in particular, OFDMA-based systems), and it has successfully addressed many system-level optimization tasks such as data offloading \cite{ho2014data}, load optimization \cite{siomina2012load,renato2017load}, antenna tilt optimization \cite{feh2013}, energy savings \cite{ho2015,renato14SPM,renato2016power}, rate optimization \cite{renato2016maxmin}, and resource muting \cite{qi2017}, to cite a few. 

\subsection{The system model}
\label{sect.system_model}
In the load coupling interference model, we divide the time and frequency grid into $K\in\Natural$ units called resource blocks. Users assigned to the same base station are not allowed to share resource blocks, but intercell interference is present because different base stations can concurrently allocate the same portion of the spectrum to serve their users. The sets of $N$ users and $M$ base stations in the network are denoted by $\setn:=\{1,\ldots,N\}$ and $\setm:=\{1,\ldots,M\}$, respectively. The set $\setn_i\subset \setn$, assumed nonempty, is the set of users connected to base station $i\in\setm$. The pathloss between base station $i\in\setm$ and user $j\in\setn$ is given by $g_{i,j}\in\real_{++}$. The vector of transmit power per resource block and the load vector are given by, respectively, $\signal{p}=[p_1,\ldots,p_M]^t\in\real_{++}^M$ and $\signal{x}=[x_1,\ldots,x_M]^t\in\real_{++}^M$, where the $i$th coordinate of these vectors correspond to the power per resource block or the load at base station $i\in\setm$. Here, load is defined to be the fraction of resource blocks that a base station uses for data transmission. Note that this model assumes uniform transmit power per resource block, and it also assumes that all resource blocks experience the same (long-term) pathloss. In this model, the downlink achievable rate of a resource block assigned by base station $i\in\setm$ to user $j\in\setn$ is given by \cite{Majewski2010,renato14SPM,renato2016power,siomina2012load,ho2014data,ho2015,feh2013}:
\begin{align*}
\omega_{i,j}(\signal{x},\signal{p})=B\log_2\left(1+\dfrac{p_i g_{i,j}}{\sum_{k\in\setm\backslash\{i\}}{x}_k p_k g_{k,j}+\sigma^2}\right),
\end{align*}
where $\sigma^2\in\real_{++}$ is the noise per resource block and $B\in\real_{++}$ is the bandwidth of each resource block. By denoting by $d_j\in\real_{++}$ the data rate requested by user $j\in\setn$, if we fix the power allocation $\signal{p}\in\real_{++}^M$, the load at the base stations is obtained by computing the fixed point (if it exists) of the continuous positive concave mapping given by   \cite{Majewski2010,siomina12,renato14SPM,feh2013}:
\begin{align}
\label{eq.load_mapping}
T:\real^M_{+}\to\real^M_{++}:\signal{x}\mapsto [ t^{(1)}(\signal{x}),\ldots,  t^{(M)}(\signal{x})]^t,
\end{align}
where, for each $i\in\setm$,
\begin{align*}
t^{(i)}:\real^M_+\to\real_{++}:\signal{x}\mapsto\sum_{j\in\setn_i} \dfrac{d_j}{K\omega_{i,j}(\signal{x},\signal{p})}.
\end{align*}
(For a formal proof of concavity of the above functions, we refer the readers to, for example, \cite[Theorem~2]{siomina12} and \cite[p. 29]{renato14SPM}.) Note that we allow the components of the load vector to take values greater than one. In practice, the load of a base station cannot exceed the value one (otherwise the base station would be transmitting with more resources than available in the system), but knowledge of these values during the planning stage of networks is useful to rank base stations according to their unserved traffic demand \cite{siomina12}. 

Instead of computing the load for a given power allocation, we can also compute the power allocation inducing a given load, and we call this problem \emph{the reverse problem}. As originally demonstrated in \cite{ho2015}, the reverse problem is important, for example, in the design of energy-efficient networks. That study has also proved that there exists a mapping in  $\mathcal{F}_\mathrm{SI}$ having as its fixed point the solution to the reverse problem, and a possible mapping has been derived in \cite{renato2016}. In more detail, given a desired load $\signal{x}\in\real_{++}^M$, the power vector $\signal{p}\in\real_{++}^M$ inducing this load, if it exists, is the fixed point of the continuous positive concave mapping given by \cite{renato2016,renato2016power,renato2016maxmin}:
\begin{align} 
\label{eq.power_mapping}
H:\real_{+}^M\to\real_{++}^M:\signal{p}\mapsto [{h}^{(1)}(\signal{p}),\ldots,{h}^{(M)}(\signal{p})], 
\end{align}
where, for each $i\in\setm$ and every $\signal{p}\in\real_{+}^M$,
\begin{align*}
{h}^{(i)}(\signal{p}):=\begin{cases}
\dfrac{p_i}{x_i} \sum_{j\in\setn_i} \dfrac{d_j}{K\omega_{i,j}(\signal{x},\signal{p})},
\quad \mathrm{if}~~ p_i\ne 0  \\ 
\sum\limits_{j\in\setn_i} \dfrac{d_j\ln 2}{KBg_{i,j}{x}_i}\left(\sum\limits_{k\in\setm\backslash\{i\}}{x}_k p_k g_{k,j}+\sigma^2\right), \\ \qquad\qquad \mathrm{otherwise}.
\end{cases}
\end{align*}

\subsection{Feasibility analysis}
\label{sect.feasibility_load_model}
We now turn our attention to the important problem of feasibility analysis of networks, which has been the focus of many previous studies \cite{Majewski2010,siomina12,renato2016,ho2015,renato14SPM}. In mathematical terms, the task is to identify conditions that guarantee the existence of the fixed point of the mapping $T$ in \refeq{eq.load_mapping} or the mapping $H$ in \refeq{eq.power_mapping}. Here we show that many results in the literature emerge as corollaries of Proposition~\ref{proposition.existence}, which can also be used in the analysis of improved models for which existing results are not applicable.

For brevity, we focus on the load mapping $T$ in \refeq{eq.load_mapping}. Using Remark~\ref{remark.example_asym}, we verify that the asymptotic mapping associated with $T$ is given by 
\begin{align}
\label{eq.load_mapping_asymptotic}
T_\infty:\real_{+}^M\to\real_{+}^M:\signal{x}\mapsto \mathrm{diag}(\signal{p})^{-1}\signal{M}\mathrm{diag}(\signal{p})\signal{x},
\end{align}
where $\mathrm{diag}(\signal{p})\in\real_+^{M\times M}$ is a diagonal matrix with diagonal elements given by the power vector $\signal{p}$, and the component $[\signal{M}]_{i,k}$ of the $i$th row and $k$th column of the matrix $\signal{M}\in\real^{M\times M}_+$ is given by
\begin{align*}
[\signal{M}]_{i,k} = \begin{cases}
	0, &\mathrm{if}~~ i= k \\
	\sum_{j \in \mathcal{N}_i} \dfrac{\mathrm{ln}(2)  d_j g_{k,j}}{KB g_{i,j}} & \mathrm{otherwise}.
	\end{cases}
\end{align*}
By Proposition~\ref{proposition.existence}, we know that $T$ has a fixed point if and only if $\rho(T_\infty)<1$. The asymptotic mapping $T_\infty$ is the linear mapping shown in \refeq{eq.load_mapping_asymptotic}, so the spectral radius $\rho(T_\infty)$ of the mapping $T_\infty$  and the spectral radius of the matrix $\mathrm{diag}(\signal{p})^{-1}\signal{M}\mathrm{diag}(\signal{p})$, and hence of the matrix $\signal{M}$, are the same. Therefore, we conclude that $\mathrm{Fix}(T)\neq \emptyset$ if and only if $\rho(\mathrm{diag}(\signal{p})^{-1}\signal{M}\mathrm{diag}(\signal{p}))=\rho(\signal{M})<1$ (here $\rho(\cdot)$ should be understood as the spectral radius of a matrix, in the conventional sense of linear algebra), which is exactly the result obtained in \cite{siomina12,ho2014data,renato2016} by using arguments with different levels of generality. This fact shows that Proposition~\ref{proposition.existence} unifies and generalizes all these existing results in the literature. Note that none of these existing results can be applied to more elaborate interference models in which, for example, the rate of a user is upper bounded because of the limited number of modulation and coding schemes.  These models have been considered in \cite{feh2013,renato14SPM}, and we now show that Proposition~\ref{proposition.existence} can also be used to verify feasibility. More specifically, for fixed $\signal{p}\in\real_{++}^M$, the study in \cite{feh2013} (see also \cite{renato14SPM}) has proposed to replace the mapping $T$ in \refeq{eq.load_mapping} with the mapping
\begin{align*}
	\bar{T}:\real^M_{+}\to\real^M_{++}:\signal{x}\mapsto [ \bar{t}^{(1)}(\signal{x}),\ldots,  \bar{t}^{(M)}(\signal{x})]^t,
\end{align*}
where, for each $i\in\setm$,
\begin{align*}
\bar{t}^{(i)}:\real^M_+\to\real_{++}:\signal{x}\mapsto\sum_{j\in\setn_i} \max\left\{\dfrac{d_j}{K\omega_{i,j}(\signal{x},\signal{p})}, \dfrac{d_j}{u}\right\}
\end{align*}
and $u\in\real_+$ is the maximum rate that each resource block can achieve because of the limited choice of modulation and coding schemes. With the above modification, $\bar{T}$ is not concave,  but it is a member of $\mathcal{F}_\mathrm{SI}$ because this set is closed under positive sums and component-wise maximum \cite{yates95,martin11,renato14SPM,feh2013}. The load is now the fixed point of the mapping $\bar{T}$, and its associated asymptotic mapping is again the mapping $T_\infty$; i.e., 
	$\bar{T}_\infty:\real_{+}^M\to\real_{+}^M:\signal{x}\mapsto \mathrm{diag}(\signal{p})^{-1}\signal{M}\mathrm{diag}(\signal{p})\signal{x},$
where $\signal{M}$ is the matrix also used in \refeq{eq.load_mapping_asymptotic}.
Therefore, by Proposition~\ref{proposition.existence}, we conclude that the improved mapping $\bar{T}$ has a fixed point if and only if $\rho(\signal{M})<1$, which is the same criterion used to identify feasibility of the mapping $T$. Therefore, the existence of a load vector satisfying $\signal{\rho}=\bar{T}(\signal{\rho})$ does not depend on whether the maximum achievable rate per resource block is limited or not. Its existence only depends on the requested rates, the available bandwidth, and the pathlosses (the information required to construct the matrix $\signal{M}$). 

If we also require information about whether the fixed point of $T$ or $\bar{T}$ has every entry with value less than one, once we know that the fixed point exists, we can compute this fixed point by using any existing algorithm \cite{yates95,nuzman07,martin11,renato2016}. Alternatively, we can also use the result in Proposition~\ref{proposition.feasibility_constraints}, which only requires  the solution to Problem~\ref{problem.cond_eig}} by using the iterations in \refeq{eq.krause_iter}, for example. The advantage of the last approach is that we avoid the initial feasibility assessment. 

\subsection{Max-min utility optimization with the load coupled interference model}
\label{sect.planning}
In the analysis in Sect.~\ref{sect.feasibility_load_model}, both the users' rates and the power of base stations are given design parameters. We now consider the problem of computing the transmit power and the corresponding load in order to maximize the minimum achievable rate of the users. In the context of the load coupled interference models described above, this utility maximization problem has been addressed in \cite{renato2016power}, which has shown that, for optimality, all users should transmit with the same rate, and the load at all base stations should be set to one \cite[Proposition~2]{renato2016power}.  Therefore, as shown in that study, the utility optimization problem can be formally written as 

\begin{problem}
	\label{problem.maxmin}
	\begin{align}
	\label{eq.maxmin}
	\begin{array}{lll}
	\mathrm{maximize}_{\signal{p}, u} & u \\
	\mathrm{subject~to} & \signal{p}\in \mathrm{Fix}(uH):=\left\{\signal{p}\in\real_{+}^M~|~\signal{p}=uH(\signal{p})\right\} \\
	& \|\signal{p}\|_\infty \le \bar{p} \\
	& \signal{p}\in\real_{+}^M, u\in\real_{++},
	\end{array}
	\end{align}
	where $\|\cdot\|_\infty$ is the standard $l_\infty$ norm,  $\bar{p}\in \real_{++}$ is the allowed transmit power (per resource block) of base stations, $H$ is the mapping in \refeq{eq.power_mapping}, and $u$ is the common achievable rate of users.
\end{problem}

The above problem is a particular instance of Problem~\ref{problem.canonical}, hence all results derived in Proposition~\ref{prop.bounds} are available. In particular, Proposition~\ref{prop.bounds} provides us with not only sharp bounds, but also with the asymptotic behavior of the solution to Problem~\ref{problem.maxmin} as $\bar{p}\to 0^+$ and $\bar{p}\to \infty$. To compute the bounds in Proposition~\ref{prop.bounds}, we only need the spectral radius of $H_\infty$, and in this particular application $H_\infty$ is the original mapping $H$ in \refeq{eq.power_mapping} with the noise power set to $\sigma=0$. \footnote{In max-min utility optimization problems involving wireless networks, we do not necessarily obtain asymptotic mappings by setting the noise power to zero. See, for example, the discussion on the mapping $T$ in Sect.~\ref{sect.feasibility_load_model}. This mapping has been used in particular instances of Problem~\ref{problem.canonical} in \cite[Sect~V-A]{renato2016maxmin} and \cite{siomina2015b}.}

For a concrete numerical example, we use the same dense network used to produce \cite[Fig.~2]{renato2016maxmin}. Briefly, we consider a stadium with 43,800 simultaneously active users and 142 micro base stations. We set the noise power spectral density to {-154~dBm/Hz}, and we vary the power budget. For brevity, we refer readers to \cite[Sect.~V-B]{renato2016maxmin} for other simulation parameters. 

Fig.~\ref{fig.utility} and Fig.~\ref{fig.energy_efficiency} show the utility (in bits/second) and the $\|\cdot\|_\infty$-transmit energy efficiency (in bits/Joule) obtained in the simulations, respectively. They also show the bounds in Proposition~\ref{prop.bounds} and the transition point in Definition~\ref{def.transition}. To construct the curves in these figures, we proceeded as follows. First, we generated $100$ values for the power budget $\bar{p}$. These values were uniformly spaced within the interval $[-30~\mathrm{dBm}, 100~\mathrm{dBm}]$. Then, for each power budget value $\bar{p}\in\real_{++}^N$ (with the power converted from dBm to Watts), we constructed a sequence $(\signal{p}_n)_{n\in\Natural}$ with the recursion $(\forall n\in\Natural)~\signal{p}_{n+1} = (1/\|H(\signal{p}_n)\|)~H(\signal{p}_n)$, where $\signal{p}_1=\signal{0}$ and $\|\cdot\|$ is the monotone norm given by $(\forall\signal{x}\in\real^N) \|\signal{x}\|=\|\signal{x}\|_\infty/\bar{p}$. By Fact~\ref{fact.nuzman} and Fact~\ref{fact.nuzman_main}, the sequence $(\signal{p}_n)_{n\in\Natural}$ converges to the optimal power allocation $\signal{p}^\star\in\real_{++}^N$, and the optimal utility $u^\star\in\real_{++}$ is given by $u^\star=1/\|H(\signal{p}^\star)\|$. In practice, we computed only the first $M$ terms of $(\signal{p}_n)_{n\in\Natural}$ for a sufficiently large $M$ (i.e., until no numerical progress could be observed), and we used $\signal{p}_M$ as the approximation of $\signal{p}^\star$. As a result, for the plots in Fig.~\ref{fig.utility} and Fig.~\ref{fig.energy_efficiency}, we show  $1/\|H(\signal{p}_M)\|$ and $1/(\bar{p}\|H(\signal{p}_M)\|)$ as the numerical approximations for, respectively, $U(\bar{p})$ and $E(\bar{p})$. In turn, to generate the graphs of the upper bounds, we computed the spectral radius with the approach in Remark~\ref{remark.comp_specrad}, and we used the simple equations in Proposition~\ref{prop.bounds}(ii)-(iii) for each value of $\bar{p}$ (since we use the $l_\infty$-norm as the norm $\|\cdot\|_b$ Proposition~\ref{prop.bounds}, we can set $\alpha=1$). Note that the upper bounds are generated by solving only one conditional eigenvalue problem, while the graphs of the functions $E$ and $U$ require the solution to a conditional eigenvalue problem for each budget $\bar{p}$ being sampled. 

As expected, in Fig.~\ref{fig.utility} and Fig.~\ref{fig.energy_efficiency} we see that all technology-agnostic properties described in the last paragraph of Sect.~\ref{sect.bounds} are present. In particular, Fig.~\ref{fig.energy_efficiency} shows that the transmit energy efficient decreases with increasing power budget $\bar{p}$, as proved in Lemma~\ref{lemma.e_monotone}. We emphasize that this result does not contradict previous studies arguing that the energy efficiency of networks is not necessarily a monotone function in the transmit power \cite{renato14SPM,li2011energy}. As discussed below Definition~\ref{def.mm_ee}, there are many nonequivalent notions of energy efficiency, and the notion of efficiency in  Definition~\ref{def.mm_ee} specialized to Problem~\ref{problem.maxmin} does not take into account, for example, the energy required to power hardware. The analysis of other notions of energy efficiency in power control problems is an important topic that we leave to a future study.

\begin{figure}
	\begin{center}
		\includegraphics[width=\columnwidth]{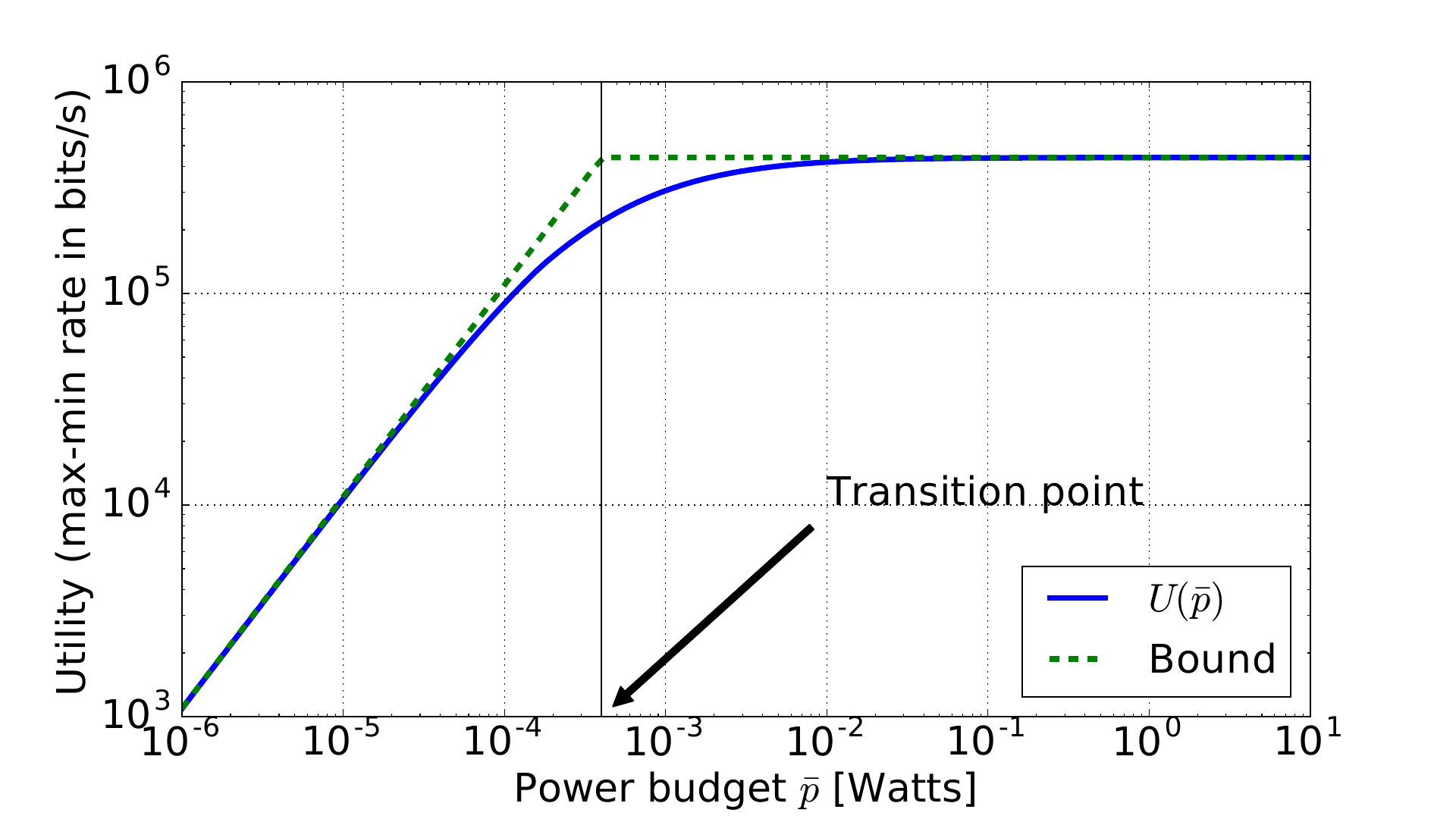}
		\caption{Network utility as a function of the power budget $\bar{p}$ for the problem described in \cite[Sect.~V-B]{renatomaxmin}.}
		\label{fig.utility}
	\end{center}
	\vspace{-.5cm}
\end{figure}

\begin{figure}
	\begin{center}
		\includegraphics[width=\columnwidth]{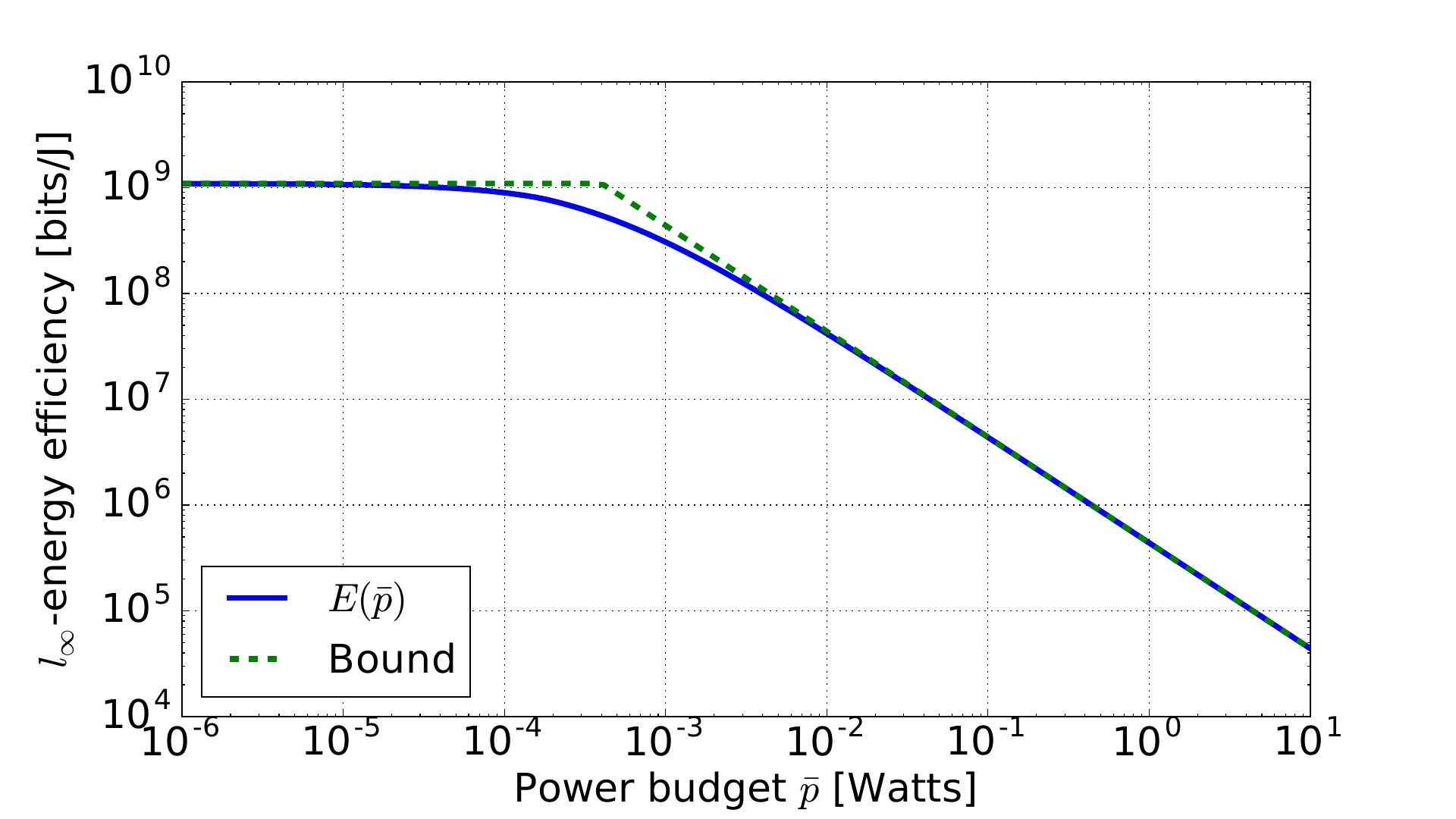}
		\caption{Transmit energy efficiency as a function of the power budget $\bar{p}$ for the problem described in \cite[Sect.~V-B]{renatomaxmin}.}
		\label{fig.energy_efficiency}
	\end{center}
		\vspace{-.5cm}
\end{figure}

\section{Summary and conclusions}

In this study, we have introduced the concept of asymptotic mappings associated with a class of mappings that are common in problems in wireless networks. As a first application of this concept, we showed that the spectral radius of asymptotic mappings can be used to explain the behavior of solutions to max-min utility optimization problems involving standard interference mappings. For example, in some max-min power control problems, the spectral radius of asymptotic mappings provides us with a power budget, called transition point, that indicates power regions in which a network is likely to be interference limited or noise limited. In these problems, gains in utility are eventually marginal if the power budget is above the transition point. In contrast, with a power budget in the region below the transition point, gains in utility are likely to be close to linear as the power budget increases. From these results, we derived analogous properties of a particular notion of transmit energy efficiency. We also showed that the spectral radius of asymptotic mappings provides us with easily computable upper bounds on the utility and the energy efficiency for any given power budget. With this result, network designers can obtain knowledge of limits of the optimal utility for a wide range of budget levels without solving a max-min utility optimization problem for different budget levels, which can be a time-consuming operation. 

As a second application, we showed that knowledge of the spectral radius of the asymptotic mapping associated with an arbitrary standard interference mapping gives a necessary and sufficient condition for the standard interference mapping to have a fixed point. This result unifies and generalizes previous tools for the feasibility analysis of networks, and it has unlocked a powerful tool for the study of networks using interference models that could not be analyzed with previous methods in the literature. Finally, we derived results providing information about the location of the fixed point of standard interference mappings, and we verified all the above contributions in concrete problems in wireless networks. 

\bibliographystyle{IEEEtran}
\bibliography{IEEEabrv,references}

\begin{thebibliography}{10}
\providecommand{\url}[1]{#1}
\csname url@rmstyle\endcsname
\providecommand{\newblock}{\relax}
\providecommand{\bibinfo}[2]{#2}
\providecommand\BIBentrySTDinterwordspacing{\spaceskip=0pt\relax}
\providecommand\BIBentryALTinterwordstretchfactor{4}
\providecommand\BIBentryALTinterwordspacing{\spaceskip=\fontdimen2\font plus
\BIBentryALTinterwordstretchfactor\fontdimen3\font minus
  \fontdimen4\font\relax}
\providecommand\BIBforeignlanguage[2]{{%
\expandafter\ifx\csname l@#1\endcsname\relax
\typeout{** WARNING: IEEEtran.bst: No hyphenation pattern has been}%
\typeout{** loaded for the language `#1'. Using the pattern for}%
\typeout{** the default language instead.}%
\else
\language=\csname l@#1\endcsname
\fi
#2}}

\bibitem{aus03}
A.~Auslender and M.~Teboulle, \emph{Asymptotic Cones and Functions in
  Optimization and Variational Inequalities}.\hskip 1em plus 0.5em minus
  0.4em\relax New York: Springer, 2003.

\bibitem{rock70}
R.~T. Rockafellar, \emph{Convex Analysis}.\hskip 1em plus 0.5em minus
  0.4em\relax Princeton university press, 1970.

\bibitem{renato2016}
R.~L.~G. Cavalcante, Y.~Shen, and S.~Sta{\'n}czak, ``Elementary properties of
  positive concave mappings with applications to network planning and
  optimization,'' \emph{{IEEE} Trans. Signal Processing}, vol.~64, no.~7, pp.
  1774--1873, April 2016.

\bibitem{renato2016maxmin}
R.~L.~G. Cavalcante, M.~Kasparick, and S.~Sta\'nczak, ``Max-min utility
  optimization in load coupled interference networks,'' \emph{{IEEE} Trans.
  Wireless Commun.}, vol.~16, no.~2, pp. 705--716, Feb. 2017.

\bibitem{yates95}
R.~D. Yates, ``A framework for uplink power control in cellular radio
  systems,'' \emph{{IEEE} J. Select. Areas Commun.}, vol.~13, no.~7, pp. pp.
  1341--1348, Sept. 1995.

\bibitem{martin11}
M.~Schubert and H.~Boche, \emph{Interference Calculus - A General Framework for
  Interference Management and Network Utility Optimization}.\hskip 1em plus
  0.5em minus 0.4em\relax Berlin: Springer, 2011.

\bibitem{slawomir09}
S.~Sta\'nczak, M.~Wiczanowski, and H.~Boche, \emph{Fundamentals of Resource
  Allocation in Wireless Networks}, 2nd~ed., ser. Foundations in Signal
  Processing, Communications and Networking, W.~Utschick, H.~Boche, and
  R.~Mathar, Eds.\hskip 1em plus 0.5em minus 0.4em\relax Berlin Heidelberg:
  Springer, 2009.

\bibitem{renato2016power}
R.~L.~G. Cavalcante, S.~Sta\'nczak, J.~Zhang, and H.~Zhuang, ``Low complexity
  iterative algorithms for power estimation in ultra-dense load coupled
  networks,'' \emph{{IEEE} Trans. Signal Processing}, vol.~7, no.~22, pp.
  6058--6070, Nov. 2016.

\bibitem{feh2013}
A.~Fehske, H.~Klessig, J.~Voigt, and G.~Fettweis, ``Concurrent load-aware
  adjustment of user association and antenna tilts in self-organizing radio
  networks,'' \emph{{IEEE} Trans. Veh. Technol.}, no.~5, June 2013.

\bibitem{nuzman07}
C.~J. Nuzman, ``Contraction approach to power control, with non-monotonic
  applications,'' in \emph{IEEE Global Telecommunications Conference
  (GLOBECOM)}.\hskip 1em plus 0.5em minus 0.4em\relax IEEE, 2007, pp.
  5283--5287.

\bibitem{ho2015}
C.~K. Ho, D.~Yuan, L.~Lei, and S.~Sun, ``On power and load coupling in cellular
  networks for energy optimization,'' \emph{{IEEE} Trans. Wireless Commun.},
  vol.~14, no.~1, pp. 500--519, Jan. 2015.

\bibitem{qi2017}
Q.~Liao and R.~L.~G. Cavalcante, ``Improving resource efficiency with partial
  resource muting for future wireless networks,'' in \emph{IEEE International
  Conference on Wireless and Mobile Computing, Networking and Communications
  ({WiMob})}, 2017, pp. 1--8.

\bibitem{gau04}
S.~Gaubert and J.~Gunawardena, ``The {Perron-Frobenius} theorem for
  homogeneous, monotone functions,'' \emph{Transactions of the American
  Mathematical Society}, vol. 356, no.~12, pp. 4931--4950, 2004.

\bibitem{gun95}
J.~Gunawardena and M.~Keane, ``On the existence of cycle times for some
  nonexpansive maps,'' Technical Report HPL-BRIMS-95-003, Hewlett-Packard Labs,
  Tech. Rep., 1995.

\bibitem{krause1986perron}
U.~Krause, ``Perron's stability theorem for non-linear mappings,''
  \emph{Journal of Mathematical Economics}, vol.~15, no.~3, pp. 275--282, 1986.

\bibitem{krause01}
------, ``Concave {Perron--Frobenius} theory and applications,''
  \emph{Nonlinear Analysis: Theory, Methods \& Applications}, vol.~47, no.~3,
  pp. 1457--1466, 2001.

\bibitem{nussbaum1986convexity}
R.~D. Nussbaum, ``Convexity and log convexity for the spectral radius,''
  \emph{Linear Algebra and its Applications}, vol.~73, pp. 59--122, 1986.

\bibitem{lem13}
B.~Lemmens and R.~Nussbaum, \emph{Nonlinear {Perron-Frobenius} theory}.\hskip
  1em plus 0.5em minus 0.4em\relax Cambridge, UK: Cambridge University Press,
  2012.

\bibitem{krause2015}
U.~Krause, \emph{Positive dynamical systems in discrete time: theory, models,
  and applications}.\hskip 1em plus 0.5em minus 0.4em\relax Walter de Gruyter
  GmbH \& Co KG, 2015, vol.~62.

\bibitem{siomina12}
I.~Siomina and D.~Yuan, ``Analysis of cell load coupling for {LTE} network
  planning and optimization,'' \emph{{IEEE} Trans. Wireless Commun.}, vol.~11,
  no.~6, pp. 2287--2297, June 2012.

\bibitem{ho2014data}
C.~Ho, D.~Yuan, and S.~Sun, ``Data offloading in load coupled networks: A
  utility maximization framework,'' \emph{{IEEE} Trans. Wireless Commun.},
  vol.~13, no.~4, pp. 1921--1931, April 2014.

\bibitem{chiang2008power}
M.~Chiang, P.~Hande, T.~Lan, C.~W. Tan, \emph{et~al.}, ``Power control in
  wireless cellular networks,'' \emph{Foundations and Trends{\textregistered}
  in Networking}, vol.~2, no.~4, pp. 381--533, 2008.

\bibitem{Majewski2010}
K.~Majewski and M.~Koonert, ``Conservative cell load approximation for radio
  networks with {Shannon} channels and its application to {LTE} network
  planning,'' in \emph{Proc. Advanced International Conference on
  Telecommunications (AICT)}, May 2010, pp. 219 --225.

\bibitem{renato14SPM}
R.~L.~G. Cavalcante, S.~Sta\'nczak, M.~Schubert, A.~Eisenbl\"ater, and
  U.~T\"urke, ``Toward energy-efficient {5G} wireless communication
  technologies,'' \emph{{IEEE} Signal Processing Mag.}, vol.~31, no.~6, pp.
  24--34, Nov. 2014.

\bibitem{renato2017SIP}
R.~L.~G. Cavalcante and S.~Sta{\'n}czak, ``The role of asymptotic functions in
  network optimization and feasibility studies,'' in \emph{Proc. IEEE Global
  Conference on Signal and Information Processing (GlobalSIP)}, Nov. 2017, pp.
  563--567.

\bibitem{renato2017performance}
------, ``Performance limits of solutions to network utility maximization
  problems,'' \emph{arXiv:1701.06491}, 2017.

\bibitem{renatomaxmin}
R.~L.~G. Cavalcante and S.~Sta\'nczak, ``Fundamental properties of solutions to
  utility maximization problems,'' in \emph{arXiv:1610.01988}, 2016.

\bibitem{john91}
C.~R. Johnson and P.~Nylen, ``Monotonicity properties of norms,'' \emph{Linear
  Algebra and its Applications}, vol. 148, pp. 43--58, 1991.

\bibitem{leung2004}
K.~K. Leung, C.~W. Sung, W.~S. Wong, and T.-M. Lok, ``Convergence theorem for a
  general class of power-control algorithms,'' \emph{IEEE Transactions on
  Communications}, vol.~52, no.~9, pp. 1566--1574, 2004.

\bibitem{stanczak2009decentralized}
S.~Sta{\'n}czak, M.~Kaliszan, and N.~Bambos, ``Decentralized admission control
  for power-controlled wireless links,'' \emph{arXiv preprint arXiv:0907.2896},
  2009.

\bibitem{baus11}
H.~H. Bauschke and P.~L. Combettes, \emph{Convex Analysis and Monotone Operator
  Theory in Hilbert Spaces}.\hskip 1em plus 0.5em minus 0.4em\relax Springer,
  2011.

\bibitem{boche2010}
H.~Boche and M.~Schubert, ``A unifying approach to interference modeling for
  wireless networks,'' \emph{IEEE Transactions on Signal Processing}, vol.~58,
  no.~6, pp. 3282--3297, 2010.

\bibitem{sun2014}
R.~Sun and Z.-Q. Luo, ``Globally optimal joint uplink base station association
  and power control for max-min fairness,'' in \emph{IEEE International
  Conference on Acoustics, Speech and Signal Processing (ICASSP)}, 2014, pp.
  454--458.

\bibitem{tan2014wireless}
C.~W. Tan, ``Wireless network optimization by {Perron-Frobenius}
  theory.''\hskip 1em plus 0.5em minus 0.4em\relax Now Publishers Inc, 2015.

\bibitem{cai2012optimal}
D.~W. Cai, C.~W. Tan, and S.~H. Low, ``Optimal max-min fairness rate control in
  wireless networks: {Perron-Frobenius} characterization and algorithms,'' in
  \emph{Proc. IEEE INFOCOM}, 2012, pp. 648--656.

\bibitem{zheng2016}
L.~Zheng, Y.-W.~P. Hong, C.~W. Tan, C.-L. Hsieh, and C.-H. Lee, ``Wireless
  max--min utility fairness with general monotonic constraints by
  {Perron--Frobenius} theory,'' \emph{IEEE Transactions on Information Theory},
  vol.~62, no.~12, pp. 7283--7298, Dec. 2016.

\bibitem{yang1998optimal}
W.~Yang and G.~Xu, ``Optimal downlink power assignment for smart antenna
  systems,'' in \emph{IEEE International Conference on Acoustics, Speech and
  Signal Processing ({ICASSP})}, vol.~6, 1998, pp. 3337--3340.

\bibitem{siomina2015b}
I.~Siomina and D.~Yuan, ``Optimizing small-cell range in heterogeneous and
  load-coupled {LTE} networks,'' \emph{{IEEE} Trans. Veh. Technol.}, vol.~64,
  no.~5, pp. 2169--2174, May 2015.

\bibitem{li2011energy}
G.~Y. Li, Z.~Xu, C.~Xiong, C.~Yang, S.~Zhang, Y.~Chen, and S.~Xu,
  ``Energy-efficient wireless communications: tutorial, survey, and open
  issues,'' \emph{IEEE Wireless Communications}, vol.~18, no.~6, pp. 28--35,
  2011.

\bibitem{golub}
G.~H. Golub and C.~F.~V. Loan, \emph{Matrix Computations}, 3rd~ed.\hskip 1em
  plus 0.5em minus 0.4em\relax Baltimore, MD: The Johns Hopkins Univ. Press,
  1996.

\bibitem{siomina2012load}
I.~Siomina and D.~Yuan, ``Load balancing in heterogeneous {LTE}: Range
  optimization via cell offset and load-coupling characterization,'' in
  \emph{IEEE International Conference on Communications ({ICC})}, 2012, pp.
  1357--1361.

\bibitem{renato2017load}
R.~L.~G. Cavalcante and S.~Sta\'nczak, ``Peak load minimization in load coupled
  interference networks,'' in \emph{Proc. IEEE Int. Conf. Acoust., Speech
  Signal Process. (ICASSP)}, March 2017.

\end{thebibliography}

\end{document}